\documentclass{lmcs}
\pdfoutput=1

\usepackage{lastpage}
\lmcsdoi{17}{2}{16}
\lmcsheading{}{\pageref{LastPage}}{}{}%
{May~01,~2019}{May~12,~2021}{}

\usepackage[utf8]{inputenc}

\usepackage{amsmath, amsfonts,amssymb, hyperref}
\usepackage[all]{xy}
\usepackage{tikz}
\usetikzlibrary{hobby}
\usetikzlibrary{arrows.meta,graphs}
\usepackage{url,breakurl,subcaption}
\usepackage{wrapfig}
\newcommand{\corn}{C-CoRn}
\newcommand{\demph}{\textbf}
\newcommand{\ask}{\operatorname{?}}
\newcommand{\answer}{\operatorname{!}}
\newcommand{\cprop}{\texttt{Prop}}
\newcommand{\bool}{\mathbb{B}}

\newcommand{\true}{\mathrm{true}}
\newcommand{\concat}{\ensuremath{{}+ \!\! +{}}}

\newcommand{\fst}{\operatorname{fst}} 
\newcommand{\snd}{\operatorname{snd}}
\newcommand{\inl}{\operatorname{inl}}
\newcommand{\inr}{\operatorname{inr}}
\newcommand{\option}{\operatorname{opt}}
\newcommand{\some}{\operatorname{Some}}
\newcommand{\none}{\operatorname{None}}
\newcommand{\seq}{\operatorname{seq}}

\newcommand{\A}{A}
\renewcommand{\AA}{\mathcal A}
\newcommand{\Q}{Q}

\newcommand{\Sp}{\mathbb S}
\newcommand{\Os}{\mathcal O}
\newcommand{\QQ}{\mathbb Q}
\newcommand{\B}{\mathcal B}
\newcommand{\X}{\mathbf X}
\newcommand{\Y}{\mathbf Y}
\newcommand{\M}{\mathbf M}

\newcommand{\Z}{\mathbf Z}
\newcommand{\NN}{\mathbb N}
\newcommand{\RR}{\mathbb R}
\newcommand{\ZZ}{\mathbb Z}
\newcommand{\mto}{\rightrightarrows}
\newcommand{\dom}{\operatorname{dom}}
\newcommand{\size}[1]{|#1|}
\newcommand{\str}{\mathbf}
\newcommand{\FTMF}{\texttt{F2MF}}
\newcommand{\Metric}{\textsc{Metric}}
\newcommand{\Coq}{\textsc{Coq}}
\newcommand{\Incone}{\textsc{Incone}}
\newcommand{\Rlzrs}{\textsc{Rlzrs}}
\newcommand{\Mf}{\textsc{mf}}

\title{Computable analysis and notions of continuity in Coq}
\titlecomment{
{\lsuper*}This work was supported by JSPS KAKENHI Grant Numbers JP18J10407 and JP20K19744, by the Japan Society for the Promotion of Science (JSPS), Core-to-Core Program (A. Advanced Research Networks), by the ANR project
FastRelax (ANR-14-CE25-0018-01) of the French National Agency for Research and by EU-MSCA-RISE project 731143 ”Computing with Infinite Data” (CID)}
\author[F.~Steinberg]{Florian Steinberg\rsuper{a}}
\address{\lsuper{a}INRIA Saclay, Bât 650, Rue Noetzlin, 91190 Gif-sur-Yvette, France}
\author[L.~Th\'ery]{Laurent Th\'ery\rsuper{b}}
\address{\lsuper{b}INRIA Sophia Antipolis, 2004 route des Lucioles, BP 93, 06902, Sophia Antipolis Cedex, France}
\author[H.~Thies]{Holger Thies\rsuper{c}}
\address{\lsuper{c}Kyushu University, 744 Motooka, Nishi-ku, Fukuoka 819-0395, Japan}
\begin{document}
\bibliographystyle{alpha}
\maketitle
\begin{abstract}
  We give a number of formal proofs of theorems from computable analysis.
  Many of our results specify executable algorithms that work on infinite inputs by means of operating on finite approximations.
  The proofs that these algorithms are correct in the sense of computable analysis are verified in  the proof assistant \Coq{} heavily relying on the \Incone{} library for information theoretic continuity.
  This library is developed by one of the authors and the paper can be used as an introduction to it.
  \Incone{} formulates the continuity-theoretic aspects of computable analysis.
  It is designed in such a way that it can be combined with \Coq{}'s Type/Prop distinction to provide a general purpose interface for algorithmic reasoning on continuous structures and many of our results provide complete computational content.
  
  The results that provide complete computational content include that the algebraic operations and the efficient limit operator on the reals are computable, that the space of infinite sequences is isomorphic to a space of functions, compatibility of the enumeration representation of subsets of natural numbers with the abstract definition as space of functions to Sierpinski space and that continuous realizability implies sequential continuity.
  We also formalize proofs of non-computational results to support the correctness of the definitions:
  We show that the information theoretic notion of continuity used in \Incone{} is equivalent to the metric notion of continuity on Baire space.
  We give a complete comparison of the different concepts of continuity that arise from metric and represented-space structures.
  Finally we prove discontinuity of the tasks of finding the limit of a converging sequence of real numbers and selecting an element from a closed subset of the natural numbers.
\end{abstract}

\section{Introduction}\label{sec: introduction}

Computable analysis is the theory of computing on continuous structures.
Its roots are often cited as going back to Turing’s fundamental paper from 1936 in which he introduced his mathematical model of computation later known as Turing machine \cite{turing_computable_1936}.
Turing's original definitions relied on the binary representation and was flawed.
In his 1937 correction he adapted it to use a signed binary representation that is still common today \cite{turing1938computable}.
For the idea behind the corrected definition he pointed to earlier work from constructive analysis by Brouwer \cite{MR0532661}.
The theory of computable functions on the real numbers was further developed in the 1950s by Grzegorczyk and Lacombe in parallel \cite{MR0089809,MR0105357}.
Finally, Kreitz and Weihrauch extended the results to apply to more general spaces and introduced the formal framework of representations that is standard today \cite{kreitz1985theory}.

The basic idea behind computable analysis is fairly simple:
To make uncountable structures available to computation, one encodes them by infinitary objects that can still be operated on mechanically.
Most commonly infinite strings are used, but more conveniently one may use functions between discrete structures.
An example of a reasonable encoding of real numbers is to describe them by functions that provide arbitrarily accurate approximations.
The inputs and outputs of such functions can be rational numbers and thus be described by finite means.
To compute functions on the real numbers, one operates on these encodings and algorithms use a model of computation that can handle infinite inputs while remaining realistic in the sense of being implementable.

Operating on continuous data such as real numbers is necessary in many applications.
The use of software based on computable analysis for solving practically relevant problems is fairly uncommon.
Partly this is due to the strict correctness assumptions that are imposed by computable analysis.
These involve high-level mathematical concepts and require algorithms to scale to arbitrary precision.
On one hand, due to the conceptual complexity, checking correctness of programs in the sense of computable analysis can be subtle and mathematical simplicity of algorithms is valued highly.
On the other hand, for practical applications performance is a central topic.
Optimizations that rely on the specifics of the hardware are not uncommon and may fail to scale.
Algorithms from computable analysis are typically difficult to optimize for performance and thus their main area of impact are applications that require utmost reliability.

Software packages based on computable analysis require their users to keep track of the details of the semantics to maintain the high reliability that computable analysis provides in principle.
This task becomes increasingly challenging to do by hand when programs become more complicated.
For instance, when dividing a program into components to be optimized separately, each subtask has to be carefully annotated with the properties that are needed for correctness of the combined procedure.
These problems are not specific to computable analysis, however, in computable analysis understanding annotations typically requires a fair amount of mathematical background knowledge.
As a consequence, many of the developers of the more popular software packages for computable analysis have started to look into possibilities to link their software to some sort of formal verification.
Ideally, this should make it possible for mathematicians to implement their ideas without having to deal with the low level details of programming languages and relying on components that computer scientists can develop and optimize.
Formal specifications guarantee that the overall programs do not suddenly become incorrect due to an optimization that changes the semantics in a subtle way and has not been properly communicated.

This paper describes a formal development where we formulated some well-known facts from computable analysis and give computer verified correctness proofs.
For the more basic operations we proceed by first obtaining a target specification from the theoretical background provided by computable analysis and then specifying a fully executable algorithm and prove it to fulfill the specification.
For more advanced operations we argue that the desired algorithm can be constructed from any set of solutions of more basic operations and finally obtain an executable program by combining what we, or other formal developments, provided earlier.
The algorithms we produce are far from being competitive in terms of speed or memory consumption.
Their relevance is more to witness theoretical implementability or prove impossibility of such an implementation.
This means that the developed theory can safely be used as a basis for verification of software packages that operate more efficiently \cite{zbMATH01746043,konecny2008aern,balluchi2006ariadne}.
Due to recent developments in verified numerics it may also be possible to optimize the algorithms underlying our proofs quite a bit without sacrificing the additional safety provided by exclusively working inside of a proof assistant \cite{BolMel17}.

The standard references to find basic facts from computable analysis are \cite{pour-el,MR1137517,weihrauch_computable_2000}.
In \cite{Bauer:2000:RAC:933370,schroeder_phd,pauly2016topological} similar topics are presented in a form that is somewhat closer to how this paper proceeds.
This paper restates many of the basic definitions from computable analysis to outline how our formal development diverges from a traditional treatment.
However, we do not attempt to give exhaustive explanations or justify more than the differences to the traditional definitions.
We assume the reader to have enough background in computable analysis to see the reasonability of the definitions used there.
For readers without background in computable analysis let us quickly point out that the model of computation used in computable analysis is by far not the only popular model for operating on functional inputs and that differences to other models can be quite subtle.
To avoid confusion, readers from adjacent fields may want to consult sources that give a systematical comparison \cite{avigad_brattka_2014,longley2015higher}.

\subsection{\Coq{} and proofs about continuous structures and computation}

As a framework for formal verification of proofs from computable analysis, we use the proof assistant \Coq{}.
This proof assistant is type theory based, but the content of this paper should be understandable with little to no background in type theory or experience with \Coq{} for that matter.
We hope that the basic facts that do make an appearance are comprehensible from the brief description of the basic workings of \Coq{} provided in this section.
For it to be possible to link the results from this paper to the formal development we provide some additional information in Appendix~\ref{app: incone}.

\Coq{} is a proof assistant that supports mathematicians in giving fully formal proofs of their results.
Ideally, a \Coq{} development looks very similar to a mathematical paper.
It consists mostly of definitions and lemmas with some explanations and some documentation in between.
The definitions specify objects that the developer is interested in and properties he wants these objects to have.
Lemmas provide evidence that some of the objects actually have some of the properties.
In practice, formal developments fall short in readability due to the need to specify every last detail and the tendency to optimize proofs for brevity instead of readability.
However, any person that has access to a computer with a \Coq{} installation can verify the correctness of a \Coq{} development without the need to understand all of its details.
Furthermore, even without understanding the details of the proofs of lemmas defined by others, it is possible to reuse them without being afraid of making errors as \Coq{} will check all the details and prevent the user from mixing up definitions or missing details.

There is a wide variety of proof assistants to choose from.
The \Coq{} system is particularly appropriate for the purpose of this paper.
This is because it is designed to support a mathematician in maintaining and extracting computational meaning while doing mathematical proofs.
The most common way of doing so is by resorting to exclusively constructive reasoning and the mathematician who chooses to do so is rewarded by programs that can be executed directly in \Coq{}.
This provides a very high standard of reliability.
For achieving this executability \Coq{} relies on the Curry-Howard correspondence to identify proofs and programs.
Indeed, proving a lemma in \Coq{} can be seen as using a high level programming language to specify how evidence of its truth can be obtained from the given inputs.
Some of the inputs here may not be traditional input data but instead evidence that the assumptions of the lemma are fulfilled.
Thus, proving a lemma may be understood as a special case of a defining a function.
One may also invert this and follow a \Coq{} definition with a proof providing the value in the high level language instead of specifying it manually.

However, it is important to keep up the distinction between providing evidence of a property and producing a value.
Just like in mathematics, the details of a value can be important while a lemma should be chosen such that all information is contained in its statement so that it can be used independently of why it is true:
It should be such that it is never necessary to unfold the proof and the exact form of the evidence provided is irrelevant.
This means that the corresponding lemmas do not have computational content and are purely for specification purposes and is independent from whether one chooses to work constructively or not.
Such lemmas can be marked as correct and the details of their proofs can be hidden.
In \Coq{} a collection \cprop\ of values that should be considered properties is provided and meant to be used to mark parts that do not have computational content in the above sense.
The distinguishing feature is that a definition of a function cannot depend on details (i.e.\ the proof) of inputs that are marked properties.
This rule is what allows the code extraction machinery of the \Coq{} system to safely disregard these parts as non-computational.

A mathematician working with \Coq{} will quickly run into statements that should intuitively be true but cannot be proven.
An example that often causes mathematicians to be confused is functional extensionality:
For functions $f$ and $g$ of the same type mathematicians would assume that the statement \verb$(forall a, f(a) = g(a)) -> f = g$ is true, but this is not provable in \Coq{}.
Depending on the upbringing of the mathematician, another example may be the law of excluded middle.
\Coq{} allows to assume the truth of properties such as the two above by stating them as axioms.
Many mathematical developments force the truth of functional extensionality by assuming it in this way.
Other popular axioms include choice principles.
Of course, one has to make sure that the axioms are compatible with \Coq{}'s internal logic and compatible with each other.
\Coq{}'s official webpages\footnote{\url{https://github.com/coq/coq/wiki/The-Logic-of-Coq\#axioms}} list some known facts about consistencies of axioms that are often used.

\subsection{Computational content in presence of axioms}\label{sec: comp-cont}
In the previous paragraph we spoke about ``axioms'' as assumptions about the logical background theory.
This should be distinguished from local axioms such as those for algebraic structures that only apply to a local context.
The best practice would be to never assert axioms as global assumptions and always prove implications that their truth has certain consequences.
However, if certain assumptions about the logical background theory are consistently used throughout a development, stating them as axioms leads to considerably cleaner overall appearance.
In our development we do this for the law of excluded middle, for functional extensionality and for some choice principles.

In principle, assuming any consistent set of axioms that state properties is an acceptable practice, but one should be aware that the use of axioms can make preserving and extracting computational content challenging.
Depending on the nature of the axioms, the loss in executability is gradual.
More specifically, axioms that assert properties typically prevent execution inside of the \Coq{} proof assistant, but an executable program in a functional programming language can still be obtained using \Coq{}'s code extraction capabilities.
Note that it is thus a plausible assumption that any function that can be defined in \Coq{} is computable even if it relies on this kind of axiom.
This remains true if the set of axioms is inconsistent:
extracted programs are still executable but due to the inconsistency it can be proven that an executable program exhibits incomputable behavior as any specification can be proven in this case.

However, \Coq{} also allows to use axioms to introduce objects that cannot realistically be constructed.
In this case the extracted algorithm can be understood to make use of a subroutine that is not provided and if it is provably impossible to implement this missing part appropriately, the produced code should be considered useless.
As example let us consider the real numbers that are axiomatized as an archimedean closed field in \Coq{}'s standard library.
This axiomatization asserts a function $\verb$up$\colon \RR \to \ZZ$ to be available of which it is assumed that it returns the least integer bigger than its input.
No computable such function exists and as a consequence, code extracted from statements about the real numbers in the standard library is rarely of practical use.
At best one obtains algorithms whose correctness relies on availability of an algorithm implementing the $\verb$up$$ function and other assumptions such as exact operations on real numbers.
A replacement of the reals by any realistically implementable type leads to an almost guaranteed loss of correctness.
This is not to say that the axiomatization is useless for specification of algorithms.

While the \verb$up$ function may not realistically implementable, a mathematician may take the position that it still exists and having direct access to such objects can be desirable.
The existence of the $\verb$up$$ function could also be stated as a property of the real numbers by hiding it behind an existential quantifier.
For the mathematical development this would result in a lot of inconveniences where an existential quantifier has to be resolved and a uniqueness lemma has to be used.
Moreover, the very mechanism that the \Coq{} system uses to support its users in keeping the distinction between computational and non-computational content will disallow the use of the \verb$up$ function definitionally if this is done.
This system can in more general cases appear as a hurdle to users who do classical mathematics, where more liberal definitional thinking is common practice.
For instance, it generally disallows branching over properties and as a result a definition by cases is often not possible even if one can prove that one of the cases must always be true.
To do branching one needs to first find an algorithm that decides the truth by returning a Boolean and such an algorithm is at times hard to come by with.
More recent developments in \Coq{}'s community for formalization of results from analysis take an even clearer stance on these topics and assume the full strength $\varepsilon$-axiom to make the development of abstract mathematics more convenient \cite{JFR8124}.

\subsection{Related work}
Verification of real and numerical analysis is an active field of research and several projects in \Coq{} and other proof assistants exist.
We do not attempt to give an overview here but instead refer to Boldo, Lelay and Melquiond's survey paper \cite{BLM16}.
Let us still briefly mention the projects that are most relevant for our formalization.

As described in the previous section, the definition of the reals in the standard library makes it difficult to strictly separate computational and non-computational content.
A solution for the real numbers that neither gives up on convenience nor interferes with the code extraction capabilities is to treat non-computational assumptions such as the $\verb$up$$ function as parameters of a real number structure and is currently under development as ``mathcomp analysis'' \cite{affeldt:hal-02463336}.
However, there exists a vast body of work building on the classical axiomatization of the reals in the standard library.
A good example for this is the Coquelicot library \cite{boldo2015coquelicot}, a widely used library for real analysis that is conservative over this axiomatization.
To make it possible to reuse previous results of this kind, we rely on the axiomatization from the standard library for the present work and manually avoid improper use of the logical strength provided by this axiomatization.

A quite different approach is taken by the \corn{} library \cite{cruz2004c}.
The library defines a fully computational formalization of real numbers.
It provides a wide range of results about functions on real numbers and some about operators on function spaces and includes an exhaustive treatment of metric spaces and uniformly continuous functions between metric spaces \cite{conorPhD}.
Its design follows the development of constructive analysis by Bishop and Bridges \cite{bishop2012constructive}.
    Our treatment of the real numbers rarely goes beyond what can already be found in \corn\ and many parts are inspired by it.
    This said, it should also be noted that the constructive nature may make the \corn\ library and the publications related to it difficult to access for some classically trained mathematicians.

    In contrast to that, the \Incone{} library follows the usual approach that is taken in computable analysis:
    The mathematical background theory is developed classically and the algorithmic content is considered extra information about data representation that should follow the mathematical understanding.
That is, it distinguishes between computational results in the form of algorithms and their correctness proofs. 
A similar approach is taken in recent work in verified numerics where an abstract mathematical theory is developed in a first step and computational content provided in an additional step by using the mathematical libraries to prove floating-point algorithms correct \cite{BolMel17}.
We hope that the ``backward approach'' taken by the \Incone{} library allows for synergy with such developments and complement the ``forwards approach'' of working completely constructively as is done in developments such as \corn.
\subsection{Structure of the paper and main results} 

The main contributions of this paper are described in Section \ref{sec: metric spaces} with some exceptions that appear earlier in Section \ref{sec: represented spaces}.
All theorems, propositions and lemmas in this paper have been formally proven in \Coq{} and have explicit pointers to their name in the \Incone{} library.
However, before we get to the core topics, some discussion of background is necessary.
Section \ref{sec: introduction} contains the introduction, pointers to literature for reading up on computable analysis and a short description of some aspects of \Coq{} that are relevant for understanding some of the more important remarks that directly address the formal development.

In Section \ref{sec: operators} we introduce the concept of continuity of partial operators on Baire space.
As a preparation for a proper treatment of partiality in \Coq{}, we introduce basic concepts from the theory of multivalued functions, a formalization of which is provided by \Incone{}'s sublibrary \Mf{} for specification of functions through relations.
We then discuss how to capture partial computable functions and operators in \Coq{} by relying on a modification of the fuel based approach to diverging computation in type theories.
This construction is one of the core concepts that is revisited many times throughout the rest of the paper.
The second part gives an information theoretic description of continuity on Baire space and an overview over the formalization of this notion in the \Incone{} library.
The third part presents the universal that the \Incone{} library uses to implement the function space construction from computable analysis.

Section \ref{sec: represented spaces} deals with the basic concepts from computable analysis, explains how they are realized in \Incone{} and introduces the real numbers as an example that is used through the rest of the section.
The first part of the section explicitly describes how a few of the simple type constructions like products are automated in the library.
Arithmetic operations on the real numbers are used as a basic example to demonstrate the workings of these constructions.
The second part describes how spaces of sequences can be constructed and considers point-wise operations on spaces of sequences and the limit operator on the real numbers as concrete examples.
From a category theoretical point of view the space of sequences is a countably infinite product of a space with itself.
As such it is of particular interest as the existence of countably infinite products is only guaranteed in the case where all continuous functions are considered as morphisms and may fail to exist if one restricts to computable ones.
The final third part builds exponentials using a construction that is known to work for both these categories.
It presents a formal proof that the space of sequences can be recovered as an exponential.

The final section (Section \ref{sec: metric spaces}) starts with a brief description of the metric library and a comparison to other formalizations that have a similar purpose.
The first part presents a formal proof that information theoretic notion of continuity that the \Incone{} library uses internally is equivalent to the more traditional approach of equipping Baire space with an appropriate metric.
The second part presents formal proofs about the relation of different concepts of continuity in metric spaces and represented spaces that can be constructed from such.
The final part introduces Sierpinski space as a space that can be used to abstractly reason about open and closed subsets of represented spaces.
For the concrete case of the natural numbers, the spaces of open and closed sets can be given more concrete representation that encode such sets by enumeration of themselves or their complement.
These concrete representations are proven equivalent to the abstract representations and the equivalence is used to prove that the task of selecting an element of a closed set does not have a continuous solution.

The formal proofs that we consider related to this publication and part of its main contributions are that the space of infinite sequences in a space isomorphic to a space of functions (Theorem \ref{resu: sig_iso_fun}), that the algebraic operations and a limit operator for fast-converging Cauchy sequences on the reals are computable (Examples \ref{resu: R operations} and \ref{resu: R limits}), compatibility of the enumeration representation of subsets of natural numbers with the abstract definition of the space of open subsets of the natural numbers (Theorem~\ref{resu: ON_iso_Onat}), and that continuous realizability implies sequential continuity (Theorem \ref{resu: cont_scnt}).
The previous results are fully algorithmic, but we also describe many non-computational theorems.
These include numerous specification results for the constructions in the \Incone{} library (in particular Theorems \ref{resu: FU_cont}, \ref{resu: D_spec}, \ref{resu: eval_cont} and Proposition \ref{resu: cprd}), a proof that the information theoretic notion of continuity used in the library is equivalent to the metric notion of continuity on Baire space (Theorem \ref{resu: cont_cont}), a complete comparison of the different concepts of continuity that arise from metric and represented-space structures (Theorems \ref{resu: scnt_mscnt} and \ref{resu: cont_mcont}) and the discontinuity of the unrestricted limit operator on the real numbers (Example \ref{resu: R limits}) and the task of selecting an element of a closed subset of the natural numbers (Theorem \ref{resu: CN_not_cont}).

\section{Multifunctions and partial operators on Baire space}\label{sec: operators}
Baire space $\NN^\NN$ is the space of all total functions from natural numbers to natural numbers.
Baire space comes with the structure of a topological space and a notion of computability of its elements.
Computable analysis transfers the computability and topological structure of Baire space to more general spaces by means of encodings that are called representations.
Before we go into detail about how this can be done, this chapter describes the structure on Baire space that we need.
We will be interested in functions on Baire space and for such the natural numbers appear in a number of different roles.
To avoid confusion, one may be tempted to name the distinct copies of the natural numbers differently.
We will do so and name the set of inputs $\Q$ for questions and the outputs $\A$ for answers so that Baire space takes the form $A^Q$.
This naming is chosen such that it matches the use of Baire space in computable analysis and to avoid coding it can often be convenient to chose $\Q$ and $\A$ to not be the natural numbers but something that can sufficiently concretely be encoded by natural numbers.

Motivated by the above we call a space of the form $\B = \A^\Q$ a \demph{naming space} if $\Q$ and $\A$ are countable and non-empty.
Here, the letter $\B$ is chosen because these spaces take the role that is traditionally taken by Baire space in computable analysis.
The phrase ``naming space'' reflects how Baire space, and thus in our treatment also any of these spaces, is used in computable analysis.
Classically the assumptions about $\Q$ and $\A$ imply that these sets are either finite or bijectively related to the natural numbers.
Constructively this need not be true, or at least depends on the notion of countability that is used.
Indeed, if computability considerations come in, that is, if the surjection whose existence is guaranteed by the countability is considered an encoding of the elements of the set by natural numbers, more care has to be taken.
The critical reader may in the following replace any occurrence of $\Q$ and $\A$ and their dashed variants by $\NN$.
In the applications that we look at, these substitutions can be carried out by hand.

We will mainly be interested in partial functions between naming spaces.
While \Coq{} comes with a native concept of a function, for \Coq{} to accept a definition of a function, the function must be total.
Thus, for our purposes, simply relying on this concept is not an option.
There are several ways to model partial functions in \Coq{}.
The most common one is to replace the target space of a function by a different space that has an additional value that stands for being undefined.
Doing this naively works well in a purely mathematical setting, but quickly leads to problems when computability considerations come in.
Another way to model partiality in \Coq{} is to rely on the dependent type system.
That is, an input of a partial function is a pair consisting of the input together with a proof that the input is from the domain of the partial function.
Such a treatment is the most natural one, but heavy use of \Coq{}'s dependent type system brings its own disadvantages.
For our purposes we decided it to be best to use aspects of the theory of multivalued functions to capture partiality.

Multifunctions are a very popular tool for specification and classification of problems in computable analysis \cite{brattka_et_al:DR:2016:5686,brattka2017weihrauch,brattka2011effective, brattka2012closed,Pauly2018}.
Within this field, multifunctions form a topic of research of their own \cite{pauly2012multi,Pauly2013RelativeCA}.
This is not to say that this concept was invented for computable analysis, multifunctions also have applications in computational complexity, in particular the theory of promise problems and non-deterministic computation \cite{SELMAN1994357,10.1007/978-3-642-39053-1_1}, and even in the treatment of non-smooth and non-linear problems in functional analysis \cite{10.2307/2371832,deimling1992multivalued}.
A \demph{multivalued function} $F \colon S \mto T$ assigns to each $s\in S$ a possibly empty subset $F(s) \subseteq T$.
While this gives $F$ the type of a relation, the intuition behind a multivalued function is different.
By contrast to relations, multivalued functions are directed and $S$ is treated as input type and $T$ as output type.
The domain of a multifunction $F$ is given by $\dom(F) := \{s \in S \mid \exists t\colon t \in F(s)\}$ and for $s\in\dom(F)$ the set $F(s)$ is non-empty and should be interpreted as the set of eligible return values.
A multivalued function is called \demph{total} if its domain is all of $S$, and \demph{singlevalued} if each $F(s)$ has at most one element.

Any multivalued function can and should be considered a specification for functions:
A function $f\colon S \to T$ fulfills the specification $F\colon S \mto T$ if for any $s\in\dom(F)$ it holds that $f(s) \in F(s)$.
In this case we say that $f$ \demph{is a choice for} $F$.
The operations on multivalued functions are chosen such that they behave well with the interpretation as specifications.
For instance, the \demph{composition} $F\circ G$ of two multivalued functions $G\colon R\mto S$ and $F\colon S \mto T$ is of type $R \mto T$ and its value sets are given by
\[ (F \circ G)(r) := \{t \in T \mid G(r) \subseteq \dom(F) \wedge \exists s\colon t \in F(s) \wedge s \in G(r) \}. \]
This should be compared to $(F\circ_R G)(r) := \{t \mid \exists s\colon t \in F(s) \wedge s \in G(r)\}$, which is what is commonly used as composition for relations.
Both the relational and the multifunction composition are associative operations.
The domain condition added in the multifunction composition is a modifier that addresses the difference in interpretations.
For the multifunction composition it is true that if $f$ is a choice for $F$ and $g$ is a choice for $G$ then $f\circ g$ is a choice for $F\circ G$, which may fail for the relational composition as illustrated in Figure~\ref{fig: comp}.
From the same figure it can be seen that the multifunction composition is not symmetric under changing the direction of the multifunctions while for relations this is the case.
The relational and the multifunction composition produce identical results if the multifunction that is applied last is total (\verb$comp_rcmp$) resp.\ the first one is singlevalued (\verb$sing_comp$).
\begin{figure}[t]
  \subcaptionbox{{\textcolor{orange}{$f$} chooses through \textcolor{red}{$F$}, \textcolor{green}{$g$} chooses through \textcolor{blue}{$G$}.
      \textcolor{purple}{$f \circ g$} does not choose through \textcolor{gray}{$F\circ_R G$}, but through $\textcolor{red}{F} \circ \textcolor{blue}{G}$ which is empty.
      $\textcolor{blue}{G}^{-1}\circ \textcolor{red}{F}^{-1} = (\textcolor{gray}{F \circ_R G})^{-1}$ is not the empty function.}
\label{fig: comp}
}[.48\textwidth]{
  \begin{tikzpicture}[thick,scale=1.5]
    \node at (0,.8) {$R$};
    \node at (1,.8) {$S$};
    \node at (2,.8) {$T$};
    \draw[color = green] (0,0) -- node[above]{\tiny $g$} (1,-.5);
    \draw[color = orange] (1,.5) -- (2,.5);
    \draw[color = orange] (1,-.5) -- node[below]{\tiny $f$} (2,-.5);
    \draw[color = purple] (0,0) -- (2,-.5);
    \node[color = purple] at (1.35,-.15) {\tiny$f \circ g$};
    \draw[fill] (0,0) circle (.05cm);
    \draw[fill] (1,.5) circle (.05cm);
    \draw[fill] (1,-.5) circle (.05cm);
    \draw[fill] (2,.5) circle (.05cm);
    \draw[fill] (2,-.5) circle (.05cm);
    \draw[rounded corners = 5pt, color = blue] (1.1,0) --  (1.1,-.675) --  (-.2,0) -- node[above]{\tiny $G$} (1.1,.675) -- (1.1,0);
    \draw[rounded corners = 3pt, color = red] (1.5,.4) --  (0.9, .4) --  (0.9,.6) -- node[above]{\tiny $F$} (2.1,.6) -- (2.1,.4) -- (1.5,.4);
    \draw[rounded corners = 3pt, color = gray] (1,.15) --  (-.075, -.1) --  (-.1,.075) -- (2.075,.6) -- (2.1,.4) -- (1,.15);
    \node[color = gray] at (1.4,.1) {\tiny$F \circ_R G$};
  \end{tikzpicture}
}
\hfill
\subcaptionbox{$f$ chooses through \textcolor{red}{$F$} which tightens \textcolor{blue}{$G$}.
Thus $f$ also chooses through \textcolor{blue}{$G$}.
\label{fig: tight}}[.48\textwidth]
{
  \begin{tikzpicture}[thick,scale=1.5]
    \draw[->] (0,0) -- (0,2);
    \draw[->] (0,0) -- (3,0);
    \node at (2.8,-.2) {$S$};
    \node at (-.2,2.2) {$T$}; 
    \node at (1.25,1.5) {\color{blue} \tiny $G$}; 
    \node at (0.15,.85) { \tiny $f$}; 
    \draw[blue] (1,1.25) to [curve through ={(1,1.25)  . . (1.9,1.2)  }]  (2,1.25);
    \draw[blue] (2,1.25) -- (2,.5);
    \draw[blue] (1,.5) to [curve through ={(1.2,.45)  . . (1.9,.4)  }] (2,.5);
    \draw[blue] (1,1.25) -- (1,.5);
    \draw[red] (.5, 1) -- node[above]{\tiny F} (.5, .75);
    \draw[red] (.5,1) to [curve through ={(1,1)  . . (1.9,1.1)  }] (2.5,1);
    \draw[red] (.5,.75) to [curve through ={(1,.65)  . . (2,1)  }] (2.5,1);
    \draw (0,.5) to [curve through ={(.5,.85)  . . (1,.8) . . (2,1.05) . . (2.5,1) }] (3,.8) ;
    
  \end{tikzpicture}
}
\caption{}
\end{figure}

Functions and partial functions are special cases of multifunctions.
A function $f\colon S \to T$ may be identified with the specification that on input $s$ it makes $f(s)$ the only eligible return value, i.e.\ with $s \mapsto \{t\in T \mid t = f(s)\}$.
The multifunction associated to a function is always total and singlevalued (\verb$F2MF_tot$ and \verb$F2MF_sing$) and assuming that $T$ is not empty and an appropriate choice principle, each total singlevalued multifunction arises in this way (\verb$fun_spec$).
We call this connection between functions and total singlevalued multifunctions the \demph{\FTMF{} correspondence}.
This construction can be extended to partial functions by assigning to $g \colon {\subseteq S} \to T$ the multifunction function
\[ s \mapsto \{t\in T \mid g(s) \text{ is defined and equals } t\}. \]
The resulting multifunction is singlevalued (\verb$PF2MF_sing$) but need not be total.
Note that we are being unspecific about how to encode partial functions here.
One may use functions to an option type as outlined previously but should keep in mind that this can be understood to imply the domain of the partial function to be decidable which is rarely the case in our applications.
It is not difficult to check that both relational and multifunction composition restrict to regular composition of functions or partial functions when considered only on singlevalued resp.\ total and singlevalued multifunctions (\texttt{F2MF\_comp} and \texttt{PF2MF\_comp}).

Any multifunction can be assigned a reverse multifunction by switching input and output.
Each property of a multifunction has a co-version that requires the same property for the reverse multifunction.
Many of the co-properties have nice characterizations for the special cases of functions.
For instance, a function is injective if and only if the associated multifunction is co-singlevalued (\verb$mfinv_inj_sing$).
For later reference we list:
\begin{lem}[\texttt{PF2MF\_cotot}]\label{resu: PF2MF_cotot}
  A partial function is surjective if and only if its associated multifunction is co-total.
\end{lem}

For multifunctions $F, G\colon S \mto T$ we say that $F$ \demph{tightens} $G$ if it is more restrictive as a specification.
This can be spelled out elementary as
\[ \dom(G) \subseteq \dom(F) \quad \text{and} \quad \forall s \in \dom(G)\colon F(s) \subseteq G(s). \]
Under weak additional assumptions this is equivalent to the statement that each choice function for $F$ is also a choice function $G$ (\verb$icf_tight$ and \verb$tight_icf$, see Figure~\ref{fig: tight}).
A function is a choice for a multifunction $F$ if and only if its associated multifunction tightens $F$ (\verb$icf_spec$).
For a partial function we say that it is a \demph{partial choice} for $F$ if its associated multifunction tightens $F$.
This can be spelled out to still mean that $f$ is a partial choice for $F$ if whenever $s \in \dom(F)$ then $f(s)$ is defined and an element of $F(s)$.
Many further properties of multifunctions are proven in the \Mf{} library and we just pick an example:
\begin{lem}[\texttt{tight\_comp}]\label{lem: tight comp}
  If $F$ tightens $F'$ and $G$ tightens $G'$ then $F\circ G$ tightens $F'\circ G'$.
\end{lem}

\subsection{Capturing the computable partial functions and relativization}\label{sec: computability}

Next we discuss another way to generate multifunctions from functions that is related to the standard way of modeling diverging computation in intuitionistic type theories.
Recall that for a type $T$ the option type $\option T$ adds a single element to $T$.
That is, each $t\in T$ corresponds to an element $\some t \in \option(T)$ and there is a single additional element $\none \in \option(T)$.
Given a function $N \colon\NN \times S \to \option T$ define a multifunction $\Phi_N\colon S \mto T$ via
\[ \Phi_N(s) := \{ t \in T \mid \exists n\colon N(n, s) = \some t\}. \]
We call the additional input $n\in\NN$ the \demph{effort parameter} and the reason for this should become apparent from the next paragraph where it takes the role of the number of steps taken by a Turing machine.

Consider the special case where $S = \NN = T$.
Given a partial computable function $f\colon{\subseteq S}\to T$ we want to argue that there exists a primitive recursive $N$ that recovers it through the $\Phi$-assignment above.
Indeed fix some Turing machine that computes $f$ and consider the function $N$ that on input $(n,s)$ returns $\some t$ if the machine on input $s$ terminates within the first $n$ time-steps and returns $t$, and $\none$ otherwise.
For this function $N$ it holds that $\Phi_N(s)$ is empty if the machine diverges and otherwise consists of the single value $f(s)$, thus $\Phi_N(s) = f$ if we identify partial functions with their induced multifunction.
Although there are technical differences, it should be clear that the core idea behind this is a version of the Kleene normal-form theorem \cite{soare1978recursively}.
The above makes the $\Phi$ assignment particularly interesting to us as any primitive recursive function is definable in \Coq{} \cite{10.1007/11541868_16}.

Even for primitive recursive $N$, the multifunction $\Phi_N$ need neither be total nor singlevalued.
In the case where $S$ and $T$ are the natural numbers, $N(n,s):= \none$ leads $\Phi_N$ to be the empty function and $N(n,s):= \some(n)$ results in the total specification for that renders any return value eligible.
However, for any $N$, it is possible to obtain a singlevalued tightening of the form $\Phi_{N'}$:
Choose $N'$ to be the function that on input $n$ and $s$ evaluates $N$ on inputs $(0,s), \ldots, (n,s)$ and returns the first value returned by $N$ if such exists and $\none$ otherwise.
This $N'$ is \demph{monotone} in the sense that if it returns a value on some effort, then it returns the same value on all bigger efforts.
Whenever a function is monotone in this sense, the corresponding multifunction is singlevalued.
As the procedure of searching is computable, it is also true that for any computable $N$ the function $\Phi_N$ has a computable choice function.
Thus, assuming that the functions definable in \Coq{} are computable (in the sense in which this was discussed in the introduction), we may use the $\Phi$ assignment to capture computability in \Coq{}.

Since we are interested in specification of partial operators, we relativize the construction.
That is, we start from the $\Phi$ assignment and add an additional argument that reaches over Baire space.
For a fixed $\varphi$ from Baire space we thus capture computability by an oracle Turing machine relative to $\varphi$ just as the original assignment captured computability by a regular Turing machine.
As the use of the word oracle indicates, $\varphi$ need not be computable.
In case the oracle is incomputable, the functions computable relative to it may be incomputable as well.
However, if one does not consider $\varphi$ to be fixed but instead as an input to the computation, one recovers a realistic model of computation and this is the model that is at the basis of computable analysis.
We further adapt to the conventions of computable analysis by interpreting the discrete input to the oracle Turing machine as a curried function input.
That is, for a given functional input $\varphi$, an oracle Turing machine is considered to return an element of Baire space if the relativized computation with oracle $\varphi$ terminates for each possible input.
If the computation on oracle $\varphi$ diverges for any of the possible inputs, the return value is considered undefined.

To make this precise, first recall that we use the notations $\B = \A^\Q$ and $\B' = \A'^{\Q'}$.
Given a function $M: \NN \times \B \times \Q' \to \option \A'$ we consider the multifunction $F_M \colon \B \mto \B'$ defined by
\[ F_M(\varphi) := \{ \psi \in \B' \mid \forall q', \exists n\colon M(n, \varphi, q') = \some(\psi(q'))\}. \]
One may repeat the discussion for the $\Phi$ assignment above to verify that for each operator computable by an oracle Turing machine, there exists some $M$ definable in G\"{o}dels system $T$ such that $F_M$ is the multivalued function assigned to the operator computed by the oracle Turing machine.
In the same way as was done for the unrelativized case, one may force singlevaluedness.
It is also still true that for any $M$ that is computable by an oracle Turing machine this singlevalued selection of $F_M$ is computable by an oracle Turing machine.
Thus, under the assumption that all functions definable in \Coq{} are computable, and whenever all the question and answer types come with canonical bijections with the natural numbers, we may use a definition of $M$ as certificate for the computability of an operator.
In the rest of the paper, whenever we mention computability, we mean computability in this sense.

As an informal example consider sections of elements of Baire space, i.e.\ let all question and answer sets be the natural numbers so that $\B$ and $\B'$ both are $\NN^\NN$.
Consider the function $M$ that on input of an effort $n\in \NN$, some $\varphi \in \B = \NN^\NN$ and a $q'\in \Q' = \NN$ returns the value $M(n, \varphi, q') \in \option(\A') = \option (\NN)$ given by
\[ M(n, \varphi, q') := \begin{cases} \some(n) & \text{if } \varphi(n) = q' \\ \none & \text{otherwise.} \end{cases} \]
Then the values of the corresponding operator $F_M$ are given by
\[ F_M(\varphi) = \{\psi \in \NN^\NN \mid \forall q'\!,\exists a'\!\colon \varphi (a') = q' \!\land \some(a') = \some(\psi(q'))\} = \{\psi \mid \forall q'\!\colon \varphi(\psi(q')) = q'\}. \]
That is $F_M(\varphi)$ is the set of all sections of $\varphi$.
Consequentially, the domain of $F_M$ is the collection of all surjective elements of Baire space.
As a given surjective function can have several different sections, $F_M$ is not singlevalued.
However, one can make the values unique by deciding on the minimal section, where a section $\psi$ of a surjection $\varphi$ is called minimal if it is pointwise less or equal to any other section.
An $M'$ such that $F_{M'}$ is the specification of being a minimal section can be found by following the construction of a singlevalued selection for the $\Phi$ assignment given above.

A version of this example that is of particular interest for us, we have made formal:\
\begin{exa}[\texttt{examples/continuous\_search.v}]\label{ex: continuous search 1}
  Consider the case where $\A := \{0,1\}$ are the Booleans, $\Q' := \{\star\}$ is the canonical one point set and the other question and answer sets are the natural numbers.
  This means that $\B = \{0,1\}^\NN$ is Cantor space and $\B' = \NN^{\{\star\}}$ can be understood as the natural numbers by identifying $\psi \in \B'$ with $\psi(\star) \in \NN$.
  Consider the operator $F\colon {\subseteq \B} \to \B'$ that on input of some $\chi$ from Cantor space that is not the constant function returning $1$, returns the first $q' \in \NN$ such that $\chi(q') = 0$.
  Then $F = F_{M}$ if we choose for $M\colon \NN\times\B\times \{\star\} \to \option(\NN)$ the function such that $M(n,\chi,\star)$ equals $\some(k)$ if $k$ is the smallest number no bigger than $n$ such that $\chi(k) = 0$ and $\none$ if no such $k$ exists.
  Just like in the informal example above, this $M$ may be constructed by first considering the multifunction returning any zero and then moving to a singlevalued selection function.
\end{exa}

Note that the $\Phi$ and the $F$ assignments differ considerably in their interpretation of what is considered input and output.
As sketched above this is not relevant for finding a singlevalued selection.
However, it renders different operations natural for the $\Phi$ and the $F$ assignment and the natural operations for the latter tend to be more problematic than for the former.
As illustration of this, and as an additional motivation for the next section, let us briefly look into composition.

First consider the $\Phi$ assignment.
Given functions $N \colon \NN \times R \to \option S$ and $N'\colon \NN \times S \to \option T$ one may define a new function $N'\circ_{\Phi} N\colon \NN \times R \to \option T$ via
\[ N' \circ_{\Phi} N (\langle n, m\rangle, r) := \begin{cases} N'(n,t) & \text{if } N(m,r) = \some t \\
  \none & \text{otherwise,} \end{cases} \]
where $\langle\cdot,\cdot\rangle$ is the Cantor (or any standard) pairing function.
This captures the relational composition in the sense that $\Phi_{N' \circ_\Phi N} = \Phi_{N'} \circ_R \Phi_{N}$ which in turn tightens $\Phi_{N'}\circ \Phi_N$.
Under the assumption that $N$ is monotone one may further simplify the construction by replacing the pair by the maximum.
Let us try to adapt the above to the setting of the $F$ assignment.
Given $M'$ and $M$ one would most likely be interested in $F_{M'} \circ F_M$.
For simplicity let us assume that $F_{M'}$ and $F_M$ are singlevalued so that the values of the composition are given by $(F_{M'}\circ F_M)(\varphi) := F_{M'}(F_M(\varphi))$.
A straightforward adaption of the proof for the $\Phi$ assignment does not give information about this operator but instead about an operation that would be more natural to consider for functionals.
Namely given functionals $\mathcal F\colon \B \times \Q' \to \A'$ and $\mathcal G\colon \B \times \Q'' \to \Q'$ it makes a statement about the functional $(\varphi, q'') \mapsto \mathcal F(\varphi, \mathcal G(\varphi, q''))$.
The construction for functionals that corresponds to what we were originally interested in is ``functional substitution'', spells out $(\varphi, q'') \mapsto \mathcal F(\lambda q'. \mathcal G(\varphi, q'), q'')$ in the language for functionals and requires different typing of $\mathcal F$ and $\mathcal G$.

Let us thus assume we are given $M\colon \NN \times \B \times \Q' \to \option\A'$ and $M'\colon \NN \times \B'\times \Q''\to \option \A''$ and let us try to find some $M' \circ_F M \colon \NN \times \B \times \Q'' \to \option \A''$ such that $F_{M'\circ_F M}$ is a tightening of $F_{M'}\circ F_M$.
Due to the definition of the $F$ assignment, we can recover from $M$ for each fixed effort $n$ and input $\varphi \in \B$ an approximation $M(n,\varphi, \cdot)\colon \Q' \to \option(\A')$ to possible functional inputs to $M'$.
However, $M'$ expects an input of type $\Q' \to \A'$ and not of type $\Q' \to \option(\A')$ and to move from the latter to the former one has to pick a value whenever $\none$ occurs as return-value.
Without further information about $M'$ it is not clear why the choices of these values should be irrelevant.
And in particular it is not clear why the return values of $M'$ on an arbitrary extension should have anything to do with the return-value of the composition.
However, \Coq{} is consistent with functional extensionality and any functions defined  in \Coq{} without use of non-computational axioms can be evaluated in a finite amount of time.
One might thus tend to believe that if $M'$ is defined in this way, then its value on fixed discrete input only relies on a finite number of the return values of its functional input.
If the original $\varphi$ is from the domain of $F_M$ it is thus possible to choose the effort big enough for the way in which we extend to be irrelevant.
The additional information that is needed about $M'$ for being able to carry out this kind of composition concretely is effective information about its continuity as presented in the upcoming section.

\subsection{Continuity of partial operators between naming spaces}\label{sec: continuity}
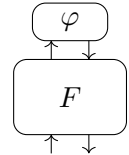
\begin{wrapfigure}{r}{.145\textwidth}
  \centering
  \vspace{-.2cm}
  \begin{tikzpicture}
    \draw[rounded corners = 5pt] (0,0) rectangle node {$F$} (1.5,1);
    \draw[->] (.5,-.25) -- (.5,0);
    \draw[->] (1,0) -- (1,-.25);
    \draw[rounded corners = 5pt] (.25,1.25) rectangle node {$\varphi$} (1.25,1.75);
    \draw[->] (.5, 1) -- (.5,1.25);
    \draw[->] (1,1.25) -- (1,1);
  \end{tikzpicture}
  \caption{A continuous operator.}\label{fig: continuous operator}
\end{wrapfigure}
This section presents an information theoretic development of a notion of continuity of operators between naming spaces.
The \Incone{} library provides proofs that the definitions presented here are equivalent to more traditional notions of continuity, but the discussion of these equivalences is postponed to Section \ref{sec: metric spaces} since it requires some background about metric spaces and topology that are not necessary for the presentation in the current section.
For the following fix some types $\Q$, $\A$, $\Q'$ and $\A'$ and set $\B:= \A^\Q$ and $\B':=\A'^{\Q'}$.

Intuitively continuity says that the values of an operator $F\colon \B \to \B'$ interpreted as functional of type $F\colon \B \times \Q' \to \A'$ only depend on finite information about the values of the functional input from $\B$ and thus can be thought of as being represented by a diagram as depicted in Figure \ref{fig: continuous operator}.
Mathematically, continuity can be described as follows:
A function $F\colon \B \to \B'$ is \demph{continuous} if for any element $\varphi$ of $\B$ and any $q'\in \Q'$ there exists a \demph{certificate}, i.e.\ a finite list $\mathbf q\in \seq \Q$ such that for any $\psi$ that coincides with $\varphi$ on $\mathbf q$ it holds that $F(\psi)(q') = F(\varphi)(q')$.
Here, two functions are said to \demph{coincide on} a finite list $\mathbf q$ if $\varphi(q) = \psi(q)$ for any $q$ that appears in $\mathbf q$.
A partial operator $F\colon {\subseteq \B} \to \B'$ is continuous if for all $\varphi \in \dom(F)$ and $q' \in \Q'$ there exists a certificate, i.e.\ a finite list $\mathbf q \in \seq \Q$ such that the above statement holds for any $\psi \in \dom(F)$ that coincides with $\varphi$ on $\mathbf q$.

\begin{exa}[\texttt{examples/continuous\_search.v}]
  Consider the function $F_0\colon \NN^\NN \to \NN^\NN$ defined by $F_0(\varphi)(n):= \varphi(n) + \varphi(0)$.
  Then for any inputs $\varphi$ and $n$ the finite list $\str q := (n, 0)$ is a certificate and thus $F_0$ is continuous.
  For the operator $F_1(\varphi)(n) := \varphi(\varphi(n))$ the list $\str q:= (\varphi(n), n)$ is appropriate.
\end{exa}
Most functions that can be explicitly defined are continuous.
For instance any function definable in G\"{o}dels system $T$ is continuous and as long as no strictly non-computational axioms are involved, it is reasonable to assume that the functions definable in \Coq{} are computable and therefore continuous (compare discussion in Section \ref{sec: comp-cont}).

\begin{wrapfigure}{r}{.145\textwidth}
  \centering
  \begin{tikzpicture}
    \draw (0,0) rectangle node {$M$} (1.5,1);
    \draw[->] (.5,-.25) -- (.5,0);
    \draw[->] (1,0) -- (1,-.25);
    \draw[rounded corners = 5pt] (.25,1.25) rectangle node {$\varphi$} (1.25,1.75);
    \draw[->] (.5, 1) -- (.5,1.25);
    \draw[->] (1,1.25) -- (1,1);
  \end{tikzpicture}
  \caption{A computable operator.}
\end{wrapfigure}
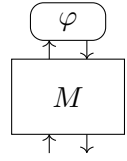

The same remains true for operators whose specification can be given as $F_M$ for some function $M\colon \NN \times \NN^\NN \times \NN \to \option \NN$ such that $M$ has a \Coq{}-definition that does not rely on non-computational axioms.
For instance for the search operator $F$ from Example \ref{ex: continuous search 1} the function $M$ can be defined in \Coq{} only using very elementary means and the list $(0, \ldots, F(\varphi)(\star))$ is a certificate for functional input $\varphi$ and discrete input $n$.
The search operator does not have a continuous total extension.
Thus, one should not expect such an extension to be definable in \Coq{} without reliance on non-computational axioms.

The definition of continuity in the \Incone{} library follows the mathematical definition given earlier mostly literally.
It superficially looks more complicated due to the use of multifunction to substitute partial functions, but the definition is chosen such that it implies a continuous multifunction to be singlevalued and does thus only really apply to partial functions.
Another difference is that instead of a list for each question the definition of continuity in \Incone{} uses a Skolem-function $L\colon \Q' \to \seq \Q$.
This switches the order of the corresponding existential and universal quantification.
Whenever an appropriate choice principle is available, these definitions are equivalent (\verb$choice_cont$).
The definition used in \Incone{} has the advantage that it allows for a fully constructive proof of the fact that the composition of continuous operators is continuous.
\begin{thm}[\texttt{cont\_comp}]\label{resu: cont_comp}
  Let $F\colon {\subseteq \B} \to \B'$ and $G\colon{\subseteq \B'} \to \B''$ be continuous partial operators.
  The operator $F \circ G \colon {\subseteq \B} \to \B''$ is continuous.
\end{thm}

\begin{proof}[summary]
  The idea behind the proof is that the certificate functions $L$ and $L'$ whose existence is guaranteed by the continuity of $F$ and $G$ can be interpreted as multivalued functions and composed relationally to obtain a certificate function for the composition of the operators.
  Furthermore, the needed relational composition can be realized constructively on the level of combining lists.
\end{proof}

As we compare different notions of continuity in the later chapters, let us briefly discuss sequential continuity on naming spaces.
Let $\B = \A^\Q$, note that $\B$ is a naming space if $\Q$ and $\A$ are countable and non-empty but for the following definition we do not need these assumptions.
An element $\varphi\in \B$ is said to be the \demph{limit} of a sequence $(\varphi_n)_{n \in \NN} \subseteq \B$ if for each fixed argument $q\in \Q$ the sequence $(\varphi_n(q))$ is eventually constantly $\varphi(q)$.
Formally
\[ \lim\nolimits_\B(\varphi_n) = \varphi \quad \iff\quad \forall q, \exists n_0, \forall n\geq n_0 \colon \varphi_n(q) = \varphi(q). \]
If a sequence in Baire space has a limit, this limit is uniquely determined (\verb$lim_sing$) and thus the above defines a partial function $\lim_\B\colon{\subseteq \B^\NN}\to \B$.

A partial operator $F\colon{\subseteq \B} \to \B'$ is called \demph{sequentially continuous} if for any $\varphi \in \dom(F)$ and any sequence $(\varphi_n)_{n \in \NN} \subseteq \dom(F)$ such that $\lim_\B(\varphi_n) = \varphi$ it also holds that $\lim_{\B'}(F(\varphi_n)) = F(\varphi)$.
It is well known that the topological structure of Baire space $\NN^\NN$ is such that sequential continuity of partial operators on Baire space is equivalent to their continuity and the \Incone{} library includes a formal proof of this and that it remains true for naming spaces.
However, this is a classical fact and constructively sequential continuity need not imply continuity, thus the library separates the equivalence into two implications.
\begin{thm}[\texttt{cont\_scnt} and \texttt{scnt\_cont}]\label{resu: cont_scnt}
  A partial operator between naming spaces is continuous if and only if it is sequentially continuous.
\end{thm}
Section \ref{sec: baire_metric} discusses how to prove that any naming space is metrizable, or more specifically it defines a metric on $\B$ from each enumeration of $\Q$.
It also presents proofs that the notions of convergence and continuity induced by this metric coincide with those given above.
Thus, the formal proofs that continuity and sequential continuity are equivalent to what is described in Section \ref{sec: continuities} imply the above theorem.
However, metric spaces use real numbers, which leads to the axioms of the real numbers appearing in the assumptions printed when inspecting the proofs.
This is even though the proofs do not use these axioms in an essential way.
Thus, the above statement is given a separate proof in the \Incone{} library.

\subsection{Construction of a universal and some of its properties}\label{sec: universal}

Recall that a naming space is a space of functions whose target and argument spaces are countable and non-empty.
A \demph{continuous universal}, or just \demph{universal}, is an assignment that for each pair $\B$ and $\B'$ of naming spaces provides another naming space $\B''$ and an operation $F_{M(\cdot)}\colon \B'' \to \B \mto \B'$ such that each of its values is continuous and for each continuous $F\colon {\subseteq\B}\to \B'$ there exists an element $\psi \in \B''$ such that $F_{M(\psi)}$ tightens $F$.
That is: A universal provides a way to code any continuous operator between naming spaces by an element of another naming space.
We call such a code, i.e.\ some $\psi$ such that $F_{M(\psi)}$ tightens $F$, an \demph{associate} of $F$ with respect to the universal or just an associate if the universal is clear from the context.
Note that the type of $\psi$ can be inferred from the universal together with the type of $F$.

Let us give some motivation for the terminology chosen here.
To justify the term ``universal'' note that replacing all the naming spaces by the set of finite binary strings and the word ``continuous'' by ``computable'' one recovers a specification that is fulfilled by the universal Turing machine.
While the construction of a universal Turing machine takes some effort, continuous universals can be chosen very simple:
Classically any naming space can be replaced by Baire space and modulo this one may use Kleene-Kreisel associateship to obtain a universal \cite{kleene1959constructivity,kreisel1959interpretation}.
A more popular variant in computable analysis is to move from Baire space to the space of infinite binary strings, i.e.\ Cantor space and use Weihrauch's $\eta$ operator \cite{weihrauch_computable_2000}.
The former of these is conceptually more well adapted to our setting and there are several excellent sources to read up about its background \cite{longley2015higher,ESCARDO2016770}.
A previous version of \Incone{}'s universal closely followed generalization of Kleene-Kreisel associateship first presented by van Oosten \cite{vanoosten2011}.
The current universal modifies that construction slightly by moving away from the idea that questions should be asked sequentially and allowing for a finite number of questions to be asked in parallel.

A mathematical description of \Incone{}'s universal can be given as follows:
for fixed naming spaces $\B = \A^\Q$ and $\B' = \A'^{\Q'}$ set $\Q'':= \seq \A \times \Q'$ and $\A'':= \seq \Q + \A'$.
That is let $\B'':= (\seq \Q + \A')^{\seq \A \times \Q'}$.
That $\B''$ is a naming space, i.e.\ that $\Q''$ and $\A''$ are countable and non-empty, follows directly from $\B$ and $\B'$ being naming spaces.
Assign to $\psi \in \B'$ the multifunction $F_{U(\psi)}\colon \B \mto \B'$ defined as follows:
$\varphi' \in F_{U(\psi)}(\varphi)$ if and only if for any $q' \in \Q'$ there exists an $N$ and a finite sequence of lists $(\str a_i)_{i\in\{1,\ldots,N\}} \subseteq \seq \A$ such that for $i<N$ it holds that $\psi(\str a_i, q') = \ask\str q$ for some $\str q \in \seq \Q$ (where $\ask$ denotes the left inclusion in the sum) and $\str a_{i+1} = \str a_i \concat (\varphi(q_1),\dots, \varphi(q_{\size{\str q}}))$ and $\psi(\str a_N,q') = \answer \varphi'(q')$ (where $\answer$ denotes the right inclusion of the sum).
The above is best understood as running a small while program:
For given functional input $\varphi$ and input $q'$, the universal attempts to extract a value from $\psi$ by first calling it on input $(\epsilon, q')$, and then branching according to the return value: if the return value is a list of questions it updates the list in the first argument with the respective answers according to $\varphi$.
If the return value is an answer, it interprets this answer as the return value of the operator (cf. Figure~\ref{fig:U}).

The notation we used strongly suggests that the universal can be specified by means of a function $U\colon \B'' \to \NN \times \B \times \Q \to \option \A'$ where the universal as described above can be recovered using the operator assignment $M \mapsto F_M$ described in Section~\ref{sec: computability} via $\psi\mapsto F_{U(\psi)}$.
Indeed, \Incone{} defines such a function \texttt U using only very elementary means.
We refrain from writing out the exact definition of $U$ here and point the interested reader to the \Incone{} library, where the definition can be printed and a formal proof that it fulfills the above specification can be found (\verb$FU_spec$).
Furthermore, $\verb$U$(\psi)$ is always monotone in the sense of the previous section (\verb$U_mon$) and in particular $F_{U(\psi)}$ is always singlevalued (\verb$FU_sing$).

\begin{figure}
  \centering
  \begin{tikzpicture}
    \draw (0,0) rectangle (4,3);
    \draw[rounded corners = 5pt] (-2.6,-.1) rectangle (5,3.1);
    \node at (4.5,1.5) {$F_{U(\psi)}$};
    \node at (3.5,1) {$U$};
    \node at (1,-.5) {$q'$};
    \draw[->] (1,-.3) -- (1,0);
    \draw[->,dotted] (1,0) -- (1,1) -- (.35,1);
    \node at (1.85,.5) {$b:=(\epsilon,q')$};
    \draw[rounded corners = 5pt] (-2.5,0) rectangle (-.5,3);
    \node at (-1.5,1.5) {$\psi$};
    \draw[->] (-.5,2) -- (0,2);
    \node at (.2,2) {$d$};
    \draw[->,dotted] (.35,2) -- (.5,2)--(.5,2.8) -- (.8,2.8);
    \node at (1.15,2.35) {$d = \ask \str q$};
    \node at (1.15,1.65) {$d = \answer a'$};
    \draw[->,dotted] (.5,2) -- (.5,1.25) -- (3,1.25) -- (3,0);
    \node at (1,2.8) {$\str q$};
    \draw[->] (0,1) -- (-.5,1);
    \node at (.2,1) {$b$};
    \draw[rounded corners = 5pt] (0,3.5) rectangle (3,5.5);
    \node at (1.5,4.5) {$\varphi$};
    \draw[->] (2,3.5) -- (2,3);
    \node at (2.6,2.75) {$\varphi(q_1)\ldots\varphi(q_{\size{\str q}})$};
    \draw[->,dotted] (2,2.5) -- (2,1) -- (1,1);
    \node at (2.9,2) {update list};
    \draw[->] (1,3) -- (1,3.5);
    \draw[->] (3,0) -- (3,-.3);
    \node at (3,-.5) {$a'$};
  \end{tikzpicture}
  \caption{The universal used in \Incone.}\label{fig:U}
\end{figure}
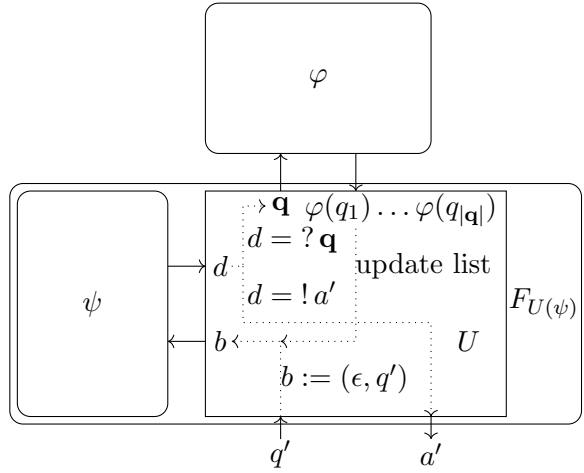

That this actually defines a universal can be separated into two statements.
Let us first argue that all operators of the form $F_{U(\psi)}$ are continuous.
We actually prove the stronger statement, that from an associate $\psi$ one can obtain a self modulating modulus of continuity for $F_{U(\psi)}$.
Let us start by introducing the notion of a modulus of continuity for a multifunction $F\colon \B \mto \B'$.
Recall from Section~\ref{sec: computability} that we decided to call a finite list a certificate for $\varphi\in\B$ and $q' \in \Q'$ if for any $\psi$ that coincides with $\varphi$ on this list the return-values of $F$ are identical.
Note that previously we assumed $F$ to be a partial function and here we talk about a multifunction, so we have to elaborate.
Call a list certificate for $\varphi$ and $q'$ if for each $\psi$ that coincides with $\varphi$ on this list and any $\varphi' \in F(\varphi)$ and $\psi' \in F(\psi)$ it holds that $\varphi'(q') = \psi'(q')$.
Note that if $F$ is singlevalued, being a certificate in this sense is equivalent to being a certificate for the corresponding multivalued function.
Furthermore note that the existence of a certificate implies that the elements of $F(\varphi)$ can only take one possible value in $q'$.
In particular, if continuity of multifunctions is defined as before, i.e.\ by requiring a certificate to exist for all inputs, the corresponding notion implies singlevaluedness and should be understood as a notion of continuity for partial functions specified by relations and not as a notion for multifunctions (cf. \cite{Pauly2013RelativeCA}).

Define the \demph{multivalued modulus of continuity} $C_F\colon \B \mto \seq \Q^{\Q'}$ of a multifunction $F\colon \B \mto \B'$ by
\[ C_F(\varphi) := \{L \colon \Q' \to \seq \Q \mid \forall q'\colon L(q') \text{ is a certificate for $\varphi$ and $q'$}\}. \]
The definition of $C_F$ makes sense for any multifunction $F\colon \B \mto \B'$ but $C_F(\varphi)$ can only be non-empty if $F(\varphi)$ has at most one element.
The continuous partial operators $F\colon {\subseteq \B} \to \B'$ can be specified as exactly those multifunctions such that the domain of $C_F$ is a super-set of the domain of $F$ (\verb$cont_spec$).
If $F$ is continuous, then we call any partial choice function $\mu\colon{\subseteq \B} \to \seq \Q^{\Q'}$ of $C_F$ a \demph{modulus of continuity} for $F$.
Note that if $\mu$ is a modulus of continuity of $F$, then $C_\mu$ has the same type as $C_F$.
Thus it makes sense to call a modulus of continuity $\mu$ for $F$ \demph{self-modulating} if it is continuous and a modulus of continuity for itself, that is if it is also a choice function for $C_\mu$.

Let $\psi$ be an associate of a partial operator $F$ with respect to the universal $U$ from above, i.e.\ let $\psi$ be such that $F_{U(\psi)}$ tightens $F$.
A self-modulating modulus for $F$ can readily be obtained from $\psi$ by tracking the queries in the evaluation of the universal.
The same can be done for the values that the universal calls the function $\psi$ on and one defines functions $U_Q$ and $U_S$ such that $F_{U_Q(\psi)}$ and $F_{U_S(\psi)}$ are the corresponding operators.
\begin{thm}[\texttt{FqM\_mod\_FU}, \texttt{FqM\_mod\_FqM} and \texttt{FqM\_mod\_FsM}]\label{resu: FU_cont}
  For any fixed $\psi$ of appropriate type the operator $F_{U_Q(\psi)}$ is a modulus of continuity for $F_{U(\psi)}$, for itself and for $F_{U_S(\psi)}$.
\end{thm}
The universal is used in \Incone{} to construct exponentials in the category of represented spaces or, more simply put, to construct spaces of functions.
The above Theorem~\ref{resu: FU_cont} in particular implies that for any $\psi$ the operator $F_{U(\psi)}$ is continuous.
Its more general statement is enough to provide what is needed to prove the evaluation procedure on the constructed space of functions to be a continuous operation.
The functions $U_Q$ and $U_S$ are of more theoretical than practical importance.
For the purpose of inspecting the evaluation of an associate \Incone{} provides a more useful function \texttt{gather\_queries} that on input of an associate $\psi$, a functional input $\varphi$, a discrete input $q'$ and an effort $n$ returns the list of all queries posed up to the $n$-th loop of the evaluation of the universal on these values (see \verb$examples/KleeneKreisel.v$ for examples).

To finish the proof that $U$ is an universal it is left to show that any partial continuous operator has an associate with respect to $U$.
\begin{thm}[\texttt{U\_universal}]\label{resu: U_universal}
  Any partial continuous operator $F\colon{\subseteq \B} \to \B'$ between naming spaces has an associate with respect to the universal $U$ described above.
  I.e.\ there exists some $\psi$ such that $F_{U(\psi)}$ tightens $F$.
\end{thm}

We do not give a lot of details for this proof here, but let us sketch the most important parts and point out some interesting details.
Let $\mu$ be any function that chooses through the multivalued modulus of continuity $C_F$ of $F$ and let $f$ be any function that chooses through $F$.
For any fixed enumeration $(q_i)_{i \in \NN}$ of $\Q$, one can attempt to define an associate $\psi$ for $F$ as follows:
On input $(\mathbf a, q')$ interpret the list $\mathbf a$ as a partial function by assuming that its elements are the return values on the first $\size {\mathbf a}$  elements mentioned in the enumeration of $\Q$, i.e. interpret it as the finite function such that $q_i \mapsto a_i$ for all $0 < i \leq \size {\str a}$.
Extend this function to a total function $\varphi_{\mathbf a}$ that is from the domain of $F$ if this is possible.
If $\mu(\varphi_{\str a}, q') \subseteq (q_1,\ldots, q_{\size{\str a}})$ then return $\answer f(\varphi_{\str a})(q')$, otherwise ask for the $(\size {\str a} +1)$-st element mentioned in the enumeration.

As $\mu$ is a modulus for $F$, evaluating this associate using the universal $U$ only results in correct return values.
That $\mu(\varphi,q')$ must be contained in some initial segment of the enumeration and the sequence of functions $\varphi_{\str a}$ converges to $\varphi$ gives hope that the iteration may often be finite.
Without further assumptions about $\mu$, however, this can not be proven and the function $\psi$ defined above might fail to be an associate of $F$.
Since it is always possible to extract a self-modulating modulus from an associate it may not be surprising that the proof can be completed if the modulus is additionally known to be self-modulating.
In the library the existence of a self-modulating modulus is proven by picking $\mu$ to be minimal with respect to subset inclusion under the additional condition that it can only return initial segments with respect to the enumeration of $\Q$.
We call such a modulus a \demph{minimal modulus} with respect to the enumeration.
\begin{lem}[\texttt{mod\_minmod}]
  A minimal modulus is always self-modulating.
\end{lem}
The reason for the indefinite article in this lemma is that the existence of a minimal modulus is not constructively provable \cite{MR966421}, and indeed the \Coq{} proof of its existence is classical and relies on a fairly strong choice principle, namely one that covers naming spaces and is thus strictly stronger than countable choice.
\begin{lem}[\texttt{exists\_minmod}]\label{lem: exists minmod}
  For any continuous partial operator there exists a minimal modulus of continuity.
\end{lem}
A minimal modulus is unique up to choice of the sequence and as it is known to be impossible to constructively prove the existence of an extensional way to obtain a modulus of continuity, the proof is inherently classical.
In particular there are computable operators on Baire space whose minimal modulus with respect to the identity as enumeration is not computable.
By contrast, from a computable associate a computable self-modulating modulus can be read of.
\Incone{} thus makes some efforts to avoid the use of the minimal modulus and instead allow to construct a self-modulating modulus from additional information about the operator.
The details of this leave the scope of this paper we do not elaborate further on this aspect.

Finally, the library defines a function $D$ that exchanges the arguments of the universal.
\begin{thm}[\texttt{D\_spec}]\label{resu: D_spec}
  For all $\varphi$ and $\psi$ it holds that $F_{U(\psi)}(\varphi) = F_{U(D\varphi)}(\psi)$.
\end{thm}
Here, the types have been purposefully omitted, details can be found in the library.
Note that, while $U(\psi)$ has the more complicated type and is interpreted as $F_{U(\psi)}$ using the operator assignment, $D$ can be directly  interpreted as a function or using the associated multifunction.
The above theorem is interesting because it is related to the Cartesian closure of the category of represented spaces (see Section \ref{sec: represented spaces} for details on the category).
However, it falls short in strength as it only considers a special case in which it is not necessary to talk about tupling of elements of naming spaces.

\begin{cor}[\texttt{FsM\_mod\_FU}, \texttt{FsM\_mod\_FsM} and \texttt{FsM\_mod\_FqM}]
  Theorem \ref{resu: FU_cont} remains true if $U_Q$ and $U_S$ are exchanged and $\psi$ and $\varphi$ are exchanged.
\end{cor}

\section{Represented spaces and continuous realizability}\label{sec: represented spaces}

Computable analysis is the theory of computation on sets of continuum cardinality.
To make such spaces available to computation, computable analysis considers encodings of such sets over Baire space $\B = \NN^\NN$.
Typically, such an encoding is understood to be a surjective partial function $\delta\colon {\subseteq \B} \mto X$ from Baire space to the set that is called a representation.
Instead of a partial function, the representation can also be considered a singlevalued multifunction and we will go back and forth between these two views seamlessly.
According to Lemma~\ref{resu: PF2MF_cotot} the requirement of being surjective translates to the corresponding multifunction being co-total.
In the formal development in \Incone{} everything is formulated using singlevalued, co-total multifunctions.
Moreover, instead of using Baire space, $\B$ is allowed to be any naming space in the sense of the previous chapter, i.e.\ any space of the form $\B = \A^\Q$ where $\A$ and $\Q$ are countable and non-empty.

A \demph{representation} of a set $X$ is defined as a naming space $\B$ together with a singlevalued, co-total multifunction $\delta\colon \B \mto X$.
Note that $\B$ can be inferred from $\delta$. This justifies the fact to just call $\delta$ the representation.
Just like in computable analysis we call those $\varphi \in \B$ such that $x \in \delta(\varphi)$ the names of $x$.
The surjectivity of $\delta$ can now be reformulated as each element of the space having at least one name and the singlevaluedness as each name uniquely identifying the element.
An alternate maybe more descriptive phrase for ``name'' would be ``description''.
As $\B$ is a naming space it is of the form $\B = \A^\Q$ for some countable and non-empty sets $\Q$ of questions and $\A$ of answers.
A useful interpretation of $\varphi$ being a name is that it provides on-demand information about $x$.
If $\varphi$ is a name of $x$ then for each question $q \in \Q$ about the abstract object $x$ the value $\varphi(q) \in \A$ can be considered a valid answer to that question.
A \demph{represented space} is a pair $\X =(X, \delta_\X)$ of a set $X$ and a representation $\delta_\X$ of $X$.
When appropriate we decorate the naming space of the representation and its question and answer sets with indices as well, i.e.\ we assume $\delta\colon\B_\X \mto X$, where $\B_\X = \A_\X^{\Q_\X}$ and $\Q_\X$ and $\A_\X$ are countable and non-empty.

As an example let us equip the real numbers with a representation.
\begin{exa}[\texttt{examples/Q\_reals.v}]\label{ex: q_reals}
  Let $X := \RR$ be the set of real numbers and pick the question and answer sets to be the rational numbers, i.e. $\Q_\RR = \A_\RR := \QQ$ and thus $\B_\RR = \QQ^ \QQ$.
  Clearly $\QQ$ is countable and non-empty so that $\B_\RR$ is a naming space.
  Let $\delta_\RR\colon {\subseteq \B_\RR} \to \RR$ be the partial function specified by
  \[ \delta_\RR(\varphi) = x \quad \iff \quad \forall \varepsilon \in \QQ, 0 < \varepsilon \implies |x - \varphi(\varepsilon)| \leq \varepsilon. \]
  Then $\delta_\RR$ is a representation of $\RR$.
  Indeed, using the axiomatization of the real numbers provided by \Coq{}'s standard library $\delta_\RR$ can be proven singlevalued and surjective (\texttt{rep\_RQ\_sing} and \texttt{rep\_RQ\_sur}) and we refer to the represented space $(\RR, \delta_\RR)$ simply as $\RR$.
\end{exa}

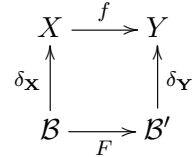
\begin{wrapfigure}{r}{0.275\textwidth}
  \begin{center}
    \hfil
    \xymatrix{
      X \ar[r]^f  & Y  \\
        \B \ar[r]_{F}  \ar[u]^{\delta_\X}     & \B' \ar[u]_{\delta_\Y}}
      \hfil
  \caption{$F\colon{\subseteq \B} \to \B'$ is a realizer of $f\colon \X \to \Y$.}
  \label{fig:realizer}
  \end{center}
\end{wrapfigure}

Computability and continuity of partial operators on naming space can be used to define computability and continuity of functions between represented spaces by means of realizers.
For represented spaces $\X$ and $\X'$, a partial operator $F\colon{\subseteq \B_\X} \to \B_{\X'}$ is a \demph{realizer} of a function $f\colon \X \to \X'$ if for each name $\varphi$ of $x$ the value $F(\varphi)$ is defined and a name of $f(x)$ (cf. Figure \ref{fig:realizer}).
A function between represented spaces is \demph{continuous} if it has a continuous realizer and \demph{computable} if it has a computable realizer.
The represented spaces form a Cartesian closed category both if the continuous functions are used as morphisms, and if the computable functions are used.
The use of the word ``continuous'' here is often contested and many say it should be reserved for the topological concept.
Others argue that the above is the correct notion and that the topological one is only a syntactic approximation of it.
It is one of the objectives of this paper to give a formal proof that the above and topological continuity have a considerable overlap in that they coincide for metric spaces.
For instance for the real numbers represented as indicated in Example~\ref{ex: q_reals} it is true that a function from $\RR \to \RR$ has a continuous realizer if and only if it is continuous in any of the more traditional ways to make sense of being continuous.
This is more generally true for functions between metric spaces and we go into detail about this in Section~\ref{sec: metric spaces}.

It is possible to extend the definition of being a realizer to allow for both the realizer and the realized function to be multivalued.
We refrain from stating the elementary definition as for the purpose of this paper it is sufficient to know that the extended definition fulfills the following specification:
\begin{lem}[\texttt{rlzr\_spec}]\label{resu: rlzr_spec}
  A multifunction $F\colon \B_\X \mto \B_{\X'}$ realizes another multifunction $f\colon \X\mto\X'$ if and only if  $\delta_{\X'}\circ F$ tightens $f\circ \delta_\X$.
\end{lem}
The above lemma can also be used backwards to express the notion of a tightening as a special case of being a realizer by using the identity function as representation (\verb$id_rlzr_tight$).
One may even further extend the definition of a realizer by dropping the requirement that the representations are singlevalued.
In the terminology of computable analysis this would mean dealing with multi-representations and while in the manipulation of discrete data the use of non-singlevalued encodings is fairly common, in computable analysis the use of multirepresentations is rare.
The above characterization does not generalize to multirepresentations.
For the full definitions we point the interested reader to the \Rlzrs{} library.

Here, we are only interested in representations and mostly in continuous, and therefore singlevalued, realizers.
However, in Section \ref{sec: closed choice} we discuss closed choice on the natural numbers as an elementary example of a multifunction between represented spaces.
We call a multifunction between represented spaces \demph{continuously realizable} if there exists a continuous realizer in the sense of the previous lemma.
Note that under the usual identification of a function with the multifunction that uniquely specifies it, continuity is a special case of continuous realizability and we sometimes use the latter to distinguish it from other notions of continuity if confusion is possible.
One reason that the use of continuously realizable multifunctions is common in computable analysis is that continuity of functions often fails for extensionality reasons.
For instance, one may formulate a multifunction corresponding to the parallel or on the space introduced in Section~\ref{sec: closed choice} as Sierpinski space to see that there exist continuously realizable functions that do not have any continuous choice function.
A special case where exactly the opposite behavior appears is that of a naming space equipped with the identity function as a representation:
a multifunction between naming spaces is continuously realizable if and only if it has a singlevalued tightening that is continuous in the sense of Section~\ref{sec: continuity} and therefore allows for a continuous choice function.

That continuity and continuous realizability is preserved under composition follows from content of the \Rlzrs{} library together with the fact that continuity of operators on Baire space is preserved under composition (Theorem \ref{resu: cont_comp}).
\begin{lem}[\texttt{comp\_cont} and \texttt{comp\_hcr}]
  The composition of continuous functions is continuous and the composition of continuously realizable multifunctions is continuously realizable.
\end{lem}

\subsection{Examples and basic constructions such as products and sums}\label{sec: simple types}

Now that we can talk about continuity and computability on the real numbers, a reasonable next step is to attempt to prove addition and multiplication computable.
Both of these functions are of type $\RR \times \RR \to \RR$ and to make sense of continuity of functions of these types we need to specify how $\RR\times\RR$ should be made a represented space.
The \Incone{} library automatically generates such a represented space $\X \times \Y$ from arbitrary represented spaces $\X$ and $\Y$ by using the query type $\Q_{\X\times\Y}:= \Q_\X + \Q_\Y$, the answer type $\A_{\X\times\Y} := \A_\X \times \A_\Y$ and the representation $\delta_{\X\times\Y}$ defined by
\[ \delta_{\X\times\Y}(\psi) = (x, y) \quad \iff \quad \delta_{\X}(\fst{} \circ \psi \circ \inl) = x \wedge \delta_{\Y}(\snd{} \circ \psi \circ \inr) = y. \]
This can be decoded as follows: A name of the pair $(x,y)$ should be a pair $(\varphi,\varphi')$ of a name for $x$ and a name for $y$.
Since the set of pairs $\B_\X \times \B_\Y$ does not have the type that we required a naming space to have, we embed it into the naming space $\B_{\X\times\Y}:=(\A_\X \times \A_\Y)^{\Q_\X + \Q_\Y}$.
There are several possible choices for $\B_{\X\times\Y}$, but for the one picked by \Incone{} the projection function $\pi\colon \B_{\X\times \Y} \to \B_\X$ can particularly conveniently be expressed by the natural operations on the question and answer spaces, namely $\pi(\psi) := \fst{}\circ \psi \circ \inl$ and the formula for the second projection is similar.
\begin{prop}[\texttt{prod\_rep\_sing}, \texttt{prod\_rep\_sur} and \texttt{prod\_uprp\_cont}]
  For any represented spaces $\X$ and $\Y$ the space $(X\times Y, \delta_{\X\times\Y})$ is a represented space and it is the product of $\X$ and $\Y$ in the category of represented spaces.
\end{prop}
Part of verifying the universal property of a product in the category of represented spaces with continuous resp.\ computable functions as morphisms is to prove the projections computable resp.\ continuous (\verb$fst_cont$ and \verb$snd_cont$).

\begin{exa}[\texttt{examples/Q\_reals.v}]\label{resu: R operations}
  Addition and multiplication of real numbers is computable (\texttt{Rplus\_cont} and \texttt{Rmult\_cont}).
  As described in more detail in Section \ref{sec: computability} this should be taken to mean that the operations are continuous and the realizers can be explicitly specified as \Coq{}-functions whose definitions contain no axioms.
  Indeed, the realizers are defined not through the more complicated operator assignment but more directly using the $\FTMF$ correspondence.
  Furthermore, their definition only uses very simple tools and the operations should therefore even be considered primitive recursive.
\end{exa}
\Incone{} additionally proves some other basic functions on product spaces computable.
Most notably it provides the possibility to glue continuous functions $f\colon \X \to \X'$ and $g\colon \Y \to \Y'$ together to obtain another continuous function $f \times g\colon \X \times \Y \to \X' \times \Y'$ (\verb$fprd_cont$).
More generally, such a construction is provided for continuously realizable multifunctions.

As another basic example of a represented space that is needed below let $I$ be any countable and non-empty set.
Set $\Q_{\mathbf I}:= \{\star\}$ and $\A_{\mathbf I} := I$.
Then the function $\delta_{\mathbf I}(\varphi) := \varphi(\star)$ makes the pair $\mathbf I:= (I, \delta_{\mathbf I})$ a represented space that is discrete in the following sense:
\begin{lem}[\texttt{cs\_id\_dscrt}]\label{resu: cs_id_dscrt}
  For any countable, non-empty set $I$ the represented space $\mathbf I$ described above is discrete in the sense that any function that has $\mathbf I$ as its domain is continuous.
Moreover, any multivalued function with $\mathbf I$ as input space is continuously realizable.
\end{lem}
In particular, the natural numbers can be assigned a discrete represented space.
We denote both the set and the represented space of natural numbers by $\NN$.

Let us briefly mention a couple of additional constructions.
The \Incone{} library proves that the represented space $\mathbf 1$ constructed from the unit type as above is a terminal object in the category of represented spaces.
It defines for each pair of represented spaces $\X$ and $\Y$ a space $\X +\Y$ that is proven to be the category-theoretical sum.
It gives a separate option type construction and proves the resulting space to be isomorphic to $\X + \mathbf 1$.
There is also an elementary construction of the space of finite lists of elements from a represented space $\X$, where asking a question about an element of $x$ results in a list of answers to question for each of the elements of the list.
\subsection{The space of infinite sequences, limits and pointwise operations}\label{sec: sequences}

Let $I$ be a countable, non-empty set and let $\X$ be a represented space.
Define a represented space $\prod_{I} \X$ whose underlying set are the functions of type $I\to X$ by setting $\Q_{\prod_{I}\X}:= I \times \Q_\X$, $\A_{\prod_I}:= \A_\X$ and
\[ (x_i) \in \delta_{\prod_I\X}(\varphi) \quad \iff\quad \forall i \colon I, x_i \in\delta_\X(q \mapsto \varphi(i,q)), \]
where $(x_i)$ is short for the function $i \mapsto x_i$.

\begin{prop}[\texttt{rep\_Iprod\_sing}, \texttt{rep\_Iprod\_sur} and \texttt{cprd\_uprp\_cont}]\label{resu: cprd}
  Let $I$ be countable and non-empty.
  Then $\prod_{I}\X := (X^I, \delta_{\prod_I \X})$ is a represented space and $\X^\omega:= \prod_\NN \X$ is a countably infinite product in the category of represented spaces and continuous functions.
\end{prop}
The use of the symbol $\omega$ instead of $\NN$ is to differentiate the space $\X^\omega$ of infinite sequences in $\X$ from the space of functions from the natural numbers to $\X$ that is discussed in the next section.
The proof that $\delta_{\prod_I\X}$ is singlevalued assumes functional extensionality and the proof of surjectivity needs a choice principle over $I$.
Since $I=\NN$ is by far the most common use-case and $I$ is assumed to be countable anyway, this usually boils down to the axiom of countable choice.
The proof of the universal property relies on stronger choice principles and the law of excluded middle.
Since the category of represented spaces with computable functions fails to have countably infinite products, the corresponding result is inherently inefficient.
Given a sequence of computable functions $f_n\colon \Y \to \X$ the function $F\colon \Y \to \X^\NN$ that glues them together need only be continuous.
Thus we only prove the existence of its realizer.
To obtain a computable realizer for $F$ one needs one algorithm that uniformly computes the $f_n$ and not just the existence of an algorithm for each $f_n$.
How far our use of axioms can be optimized in this case is difficult to tell at this point in time since the current proof uses a part of the library that has not yet been optimized in terms of axiom use.
Since it is more a sanity result than something that may actually be of use, optimizations here are not our highest priority.

An example of a partial function whose natural domain is a subset of the space of sequences is the limit operator.
Consider the multivalued function $\lim_\X\colon \X^\omega \mto \X$ where $x \in \lim_\X(x_n)$ if and only if there is a convergent sequence of names $(\varphi_n) \subseteq \B_\X$ and some $\varphi$ such that $\varphi$ is a name of $x$, each $\varphi_n$ is a name for $x_n$ and the sequence $(\varphi_n)$ converges to $\varphi$ in $\B_\X$, i.e.\ $\lim_{\B_\X}(\varphi_n) = \varphi$ where the limit in $\B_\X$ is taken point-wise as explained in Section \ref{sec: continuity}.
While the limit operator on Baire space is singlevalued, this need not be true for the limit operator on a general represented space, as can be seen at the example of Sierpinski space that is discussed in Section \ref{sec: closed choice}.
In most spaces that are relevant for numerical analysis, the limit operator is singlevalued but discontinuous.
It is often the case that computability of the limit operator can be recovered by restricting it to an appropriate set of efficiently convergent sequences.
\begin{exa}[\texttt{examples/Q\_reals.v}]\label{resu: R limits}
  The limit operator $\lim_\RR$ where $\RR$ is represented as in Example \ref{resu: R operations} is discontinuous (\texttt{lim\_not\_cont}).
  Its restriction to those sequences $(x_n)$ that are efficiently Cauchy in the sense that $|x_n - x_m| \leq 2^{-n} + 2^{-m}$ is computable (\texttt{lim\_eff\_hcr}).
\end{exa}
To be strict, the example file proves these properties with respect to the metric notion of convergence that is introduced later in Section~\ref{sec: continuities}.
To really obtain what is claimed here one has to additionally use that metric convergence and convergence in the represented space of real numbers are equivalent.
The major part of this equivalence is proven in Theorem~\ref{lem:lim-mlim} and the rest, namely that the representation of real numbers introduced in Example~\ref{ex: q_reals} is equivalent (in the sense introduced shortly) to the representation of real numbers as metric space, is also provided by the \Incone{} library.

A function $f\colon\X\to\Y$ between represented spaces is called sequentially continuous if it preserves limits, i.e.\ if $\lim_\X x_n = x$ implies that $\lim_\Y f(x_n) = f(x)$.
In general, sequential continuity need not imply continuity and this is not a matter of whether one works constructively or not.
From the point of view of a classical mathematician there exist sequentially continuous functions between represented spaces that are simply not continuously realizable.
The difference can be eliminated by making additional assumptions about the involved represented spaces.
One such assumption is referred to as admissibility and while we do not go into the details here, this condition is of major importance throughout different parts of computable analysis \cite{schroder2002extended}.
Most spaces encountered in practice are admissible, in particular the space of real numbers from Example~\ref{ex: q_reals}, all the representations of metric spaces considered in Section~\ref{sec: metric spaces}, and also the representations of hyper-spaces that are topic of Section~\ref{sec: closed choice} are admissible.
The relation between sequential continuity and continuous realizability is discussed in more detail in Section \ref{sec: continuities}.

The \Incone{} library proves some further lemmas about spaces of sequences that might be useful in applications and should thus not go unmentioned.
Two represented spaces $\X$ and $\Y$ are \demph{isomorphic}, in symbols $\X \simeq \Y$, if there exists a continuous bijection with continuous inverse.
The spaces are computably isomorphic if there exists a computable bijection with computable inverse.
In the special case where the underlying set of two represented spaces are identical we call their representations \demph{equivalent} if the identity function is an isomorphism.
While mathematicians are usually fairly liberal in identifying sets and equivalence of representations is a widely used concept, in our formal development isomorphy is by far the more common notion.
\begin{lem}[\texttt{cprd\_prd}]
  For any represented spaces $\X$, $\Y$ and countable, non-empty set $I$ it holds that
  $\prod_{I}(\X \times \Y) \simeq \big(\prod_{I}\X\big) \times \big(\prod_{I}\Y\big)$, the spaces are even computably isomorphic.
\end{lem}
The realizers are defined using very limited means and interpreted using the $\FTMF$ assignment, thus they should be considered primitive recursive.

Any function $f\colon \X \to \Y$ can be extended to a function $f^I \colon \prod_I \X \to \prod_I \Y$ that applies $f$ pointwise, i.e.\ is defined by $f^I((x_i)_{i\in I}) := (f(x_i))_{i\in I}$.
\begin{lem}[\texttt{ptw\_cont}]
  Whenever $f$ is a continuous function, then also $f^I$ is continuous.
\end{lem}
This lemma has a multivalued variant (\texttt{ptw\_hcr}).
If the realizer of $f$ can be expressed as a \Coq{}-function through the $\FTMF$ interpretation, then so can $f^I$.
The same should hold true if a realizer of $f$ can be expressed via a \Coq{} function through the operator assignment $M \mapsto F_M$ from Section~\ref{sec: computability}, although we have not carried out the details yet.

Another common situation is that an operation $*\colon \X \times \Y \to \Z$ is used to construct an operation $*^I\colon \prod_I \X \times \prod_I \Y \to \prod_I \Z$ via $(x_i) *^I (y_i) := (x_i * y_i)$.
For instance the natural operations on diverse spaces of sequences of real numbers are introduced by pointwise applying the arithmetic operations of the real numbers.
More concretely, the most common vector space structure on the space $\RR^\omega$ of sequences in the real numbers is to set $\X=\Y=\Z:=\RR$ in the above and use $+^\NN$ where $+$ is addition of real numbers.
A proof that this extension also preserves continuity can be directly obtained from the previous two lemmas.
\begin{cor}[\texttt{cptw\_op\_cont}]
   If $*$ is a continuous operation, then the corresponding pointwise operation $*^I$ is also continuous.
\end{cor}
As mentioned before, examples where these results are useful is when reasoning about pointwise addition and multiplication as operations on $\RR^\omega$.
The operation of multiplying a sequence with a scalar is also covered by first embedding the scalars into the sequences as the constant sequences and then using pointwise multiplication.

\subsection{Function spaces and infinite sequences as functions}\label{sec: function spaces}

Let $\X$ and $\Y$ be represented spaces and denote by $\Y^\X$ the collection of all continuous functions from $\X$ to $\Y$.
Recall that Section \ref{sec: universal} gave an explicit description of a continuous universal $U$.
This universal can be used to equip the space $\Y^\X$ with a representation as follows:
Let $\B_{\Y^\X}$ be the space that the universal $U$ points to for input naming spaces $\B_\X = \A_\X^{\Q_\X}$ and $\B_\Y = \A_\Y^{\Q_\Y}$.
This expands to $\Q_{\Y^\X} = \seq \A_\X \times \Q_\Y$ and $\A_{\Y^\X} = \seq \Q_\X + \A_\Y$, where $\B_{\Y^\X} = \A_{\Y^\X} ^ {\Q_{\Y^\X}}$.
Note that a function (as opposed to a partial or multifunction) is uniquely determined by each of its realizers.
Thus there exists a unique partial function $\delta_{\Y^\X}\colon \B_{\Y^\X}\mto \Y^\X$ specified by
\[ \delta_{\Y^\X}(\psi) = f \quad \iff \quad F_{U(\psi)} \text{ realizes } f. \]
Put differently, although $\delta_{\Y^\X}$ is specified as a multifunction, it is singlevalued.
Any $f \in \Y^\X$ is continuous which by definition means that there exists some continuous $F\colon{\subseteq \B_\X} \to \B_\Y$ that realizes $f$.
Because $U$ is a universal there exists some $\psi$ such that $F_{U(\psi)}$ tightens $F$ (cf.\ Theorem \ref{resu: U_universal}).
As being a realizer is preserved under tightening also $F_{U(\psi)}$ realizes $f$.
Thus $\delta_{\Y^\X}$ is surjective and co-total when considered a multifunction.

\begin{prop}[\texttt{fun\_rep\_sing} and \texttt{fun\_rep\_sur}]\label{prop:fun-rep}
  For any represented spaces $\X$ and $\Y$ the space $(\Y^\X,\delta_{\Y^\X})$ of continuous functions as defined above is a represented space.
\end{prop}
The proof of singlevaluedness assumes proof irrelevance and functional extensionality.
Here proof irrelevance is another commonly assumed axiom that asserts a \Coq{} property (\cprop) and that over the internal logic of \Coq{} is implied by the law of excluded middle.

The following theorem uses the finite products as they were constructed in Section~\ref{sec: simple types} and can be proven from the continuity properties of the universal from Theorem~\ref{resu: FU_cont}.
\begin{thm}[\texttt{eval\_cont}]\label{resu: eval_cont}
  Evaluation as operation $\Y^\X\times \X \to \Y$ is computable.
\end{thm}

The function space construction overlaps in its scope with the space of infinite sequences:
For a countable, non-empty index set $I$ and a represented space $\X$, the set underlying the space $\prod_I \X$ is the set of functions from $I$ to $\X$.
Recall what it meant for the discrete space $\mathbf I$ generated from $I$ as in Section \ref{sec: simple types} to be discrete:
Proposition \ref{resu: cs_id_dscrt} says that any function starting from $\mathbf I$ is continuous.
As a consequence, the sets underlying $\prod_I \X$ and $\X^{\mathbf I}$ are identical.
Indeed these spaces are computably isomorphic.

\begin{thm}[\texttt{sig\_iso\_fun}]\label{resu: sig_iso_fun}
  For any represented space $\X$ and any countable, non-empty set $I$ the space $\prod_{I} \X$ from the last section is computably isomorphic to the function space $\X^{\mathbf I}$, where $\mathbf I$ the discrete space constructed from $I$ as described in Section \ref{sec: simple types}.
\end{thm}
\begin{proof}[sketch]
  Recall that $\A_{\X^{\mathbf I}} = \seq \Q_{\mathbf I} + \A_\X$ and that we decided to refer to the left inclusion into this sum by $?$ and to the right inclusion by $!$.
  A realizer $T\colon \B_{\prod_I\X} \to \B_{\X^{\mathbf I}}$ that translates a name of a sequence to a name of the corresponding function can be directly specified via
  \[ T(\varphi)(\str a,q') := \begin{cases} ?(\star) & \text{if } \str a = \epsilon \\ !\varphi(a_1,q') & \text{, where } \str a = (a_1, \ldots, a_{|\str a|}). \end{cases} \]
  To see that this realizer is correct we need to prove that whenever $\varphi$ is a $\prod_I\X$-name of some $(x_i)$, then $F_{U(T(\varphi))}$ is a realizer of the function $i \mapsto x_i\colon \mathbf I \to \X$.
  The former means that for any question $q$ about $x_i$, $\varphi(i, q)$ is a valid answer.
  To check that $F_{U(T(\varphi))}$ is a realizer note that $\star \mapsto i$ is the only valid name of $i$ in $\mathbf I$ and that for the correctness of $T$ it is thus enough to argue that $F_{U(T(\varphi))}(\star \mapsto i)$ on any input $q$ returns $\varphi(i,q)$.
  Now the values of $F_{U(T(\varphi))}$ are obtained by evaluating the universal.
  If the evaluation of $U(T(\varphi))$ is started on functional input $\star\mapsto i$ and input $q$ it will first call $T(\varphi)$ on input $(\epsilon, q)$ and be returned $?(\star)$.
  This will lead $U$ to evaluate $\star \mapsto i$ in $\star$ and call $T(\varphi)$ on the updated input $((i), q)$.
  As $T(\varphi)$ returns $!\varphi(i,q)$ on this input $U$ will finish its run by returning $\varphi(i,q)$ and we have proven correctness of the translation.

  The translation in the other direction, i.e.\ constructing a sequence from a continuous function proceeds by using a variation of the realizer of evaluation.
\end{proof}
Note that the first translation can be defined by very elementary means but is specific to the details of the universal $U$.
The translation in the other direction is independent of the implementation of the universal but only relies on availability of an algorithm for evaluation.
The algorithm for evaluation, however, executes the universal which requires an unbounded search and is considerably less elementary.
More specifically, the first translation can be defined as a function that is interpreted directly or as its associated multifunction, while the second translation requires the use of the more complicated operator assignment discussed in detail in Section~\ref{sec: computability}.
Indeed, the second translation need in general not have a total continuous realizer and thus one should not expect it to be possible to give a definition of a translation as elementary as the one for the first direction.

The above theorem is a good example of a result that is usually stated as an equivalence but that was formulated as an isomorphism in our formal development.
The sets underlying $\X^\omega$ and $\X^\NN$ are considered equal from a mathematical point of view, but formally the elements of the latter are pairs of a function and a proof that this function is continuous.
The set of such pairs is still in bijection with its first components because we work in a setting where different proofs of the same property are considered equal.
More specifically the equality of these proofs is what the axiom of proof irrelevance states and over the logic of \Coq{} proof irrelevance is implied by the law of excluded middle that we regularly assume anyway.
The bijection can thus be given as the function that just drops the proof in the second component of the pair and for the inverse one lifts an arbitrary function by adding the proof obtained from the fact that $\NN$ is discrete.

Theorem~\ref{resu: sig_iso_fun} above is an important building block for obtaining the results about hyper-spaces that are presented Section~\ref{sec: closed choice}.
As a more basic example of an application of this theorem let us give some additional justification for our previous choice to use the identity function as canonical representation of a given naming space.
More specifically let us prove that the structure thus obtained is the same as considering a naming space as space of functions from its question set as discrete space to its answer set as such.
\begin{cor}
  For any naming space $\B = \A^\Q$ it holds that $(\B,\mathrm{id}_\B) \simeq \mathbf A^{\mathbf Q}$.
\end{cor}
\begin{proof}
  By the previous theorem $\mathbf A^{\mathbf \Q} \simeq \prod_{Q}\mathbf A$.
  The equivalence of the representation of $\prod_{Q} \mathbf A$ with the identity representation is straightforward to define directly.
  To be exact, $T(\varphi)(q, \star):= \varphi(q)$ continuously translates from the identity representation to that of $\prod_{Q}\mathbf A$ and $T'(\varphi')(q) := \varphi'(q, \star)$ continuously translates back.
\end{proof}
\section{Application to continuity in metric and hyper spaces}\label{sec: metric spaces}

A map $d\colon M \times M \to \RR$ is called a \demph{pseudo-metric} on a set $M$ if it returns non-negative values, is symmetric under exchange of its arguments and fulfills $d(x,x)= 0$ and the triangle inequality
\[ d(x,z) \leq d(x,y) + d(y,z). \]
It is called a \demph{metric} if it is additionally true that $d(x,y)=0$ implies $x = y$.
A pair $(M, d)$ is called a \demph{pseudo-metric space} if $d$ is a pseudo-metric on $M$ and a \demph{metric space} if $d$ is a metric.
Every pseudo-metric space comes with a topology that is generated by the open balls with respect to the pseudo-metric and therefore with notions of continuity of functions on and convergence of sequences in this space.
The notion of convergence is of particular importance since any pseudo-metric space is first-countable and thus knowing the limits of sequences is sufficient for characterizing continuity.
The notion of a metric captures the properties that one expects a notion of distance to have in a very general sense.
Metric spaces are a widely applicable tool for talking about continuity on many spaces of practical interest and a common sight in many branches of mathematics.
As such, metric spaces have received considerable attention in their formal treatment.

A definition of continuity that does not require any knowledge of topology can be given using the $\varepsilon$-$\delta$-criterion.
A function $f\colon M \to M'$ between pseudo-metric spaces $(M,d_M)$ and $(M', d_{M'})$ is called \demph{continuous in} $x \in M$ if
\[ \forall \varepsilon > 0, \exists \delta > 0, \forall y\colon d_{M} (x, y) \leq \delta \implies d_{M'}(f(x),f(y)) \leq \varepsilon. \]
Here, $\varepsilon$ and $\delta$ are real numbers but replacing them by rational numbers results in an equivalent definition.
The function is called \demph{continuous} if it is continuous in every point of $M$.
An element $x$ of a pseudo-metric or metric space $(M, d)$ is said to be the \demph{limit} of a sequence $(x_n)$ in $M$, in symbols $\lim_{(M,d)}(x_n) = x$, if
\[ \forall \varepsilon>0, \exists N,\forall n \geq N\colon d (x, x_n) \leq \varepsilon. \]
Again, $\varepsilon$ is taken to be a real number but may equivalently be restricted to be rational.
A function $f$ between pseudo-metric spaces is said to be \demph{sequentially continuous in} $x$ if for each sequence $x_n$ in $M$ with $\lim_{(M,d_M)}(x_n) = x$ it holds that $\lim_{(M',d_{M'})}(f(x_n)) = f(x)$ and \demph{sequentially continuous} if it is sequentially continuous in every point of $M$.

Note that above for metric spaces we reuse the terms that we already used for represented spaces: We give an alternative definition of what it means to be continuous, a limit and sequentially continuous.
Since we make a clear distinction of whether the spaces we operate on are metric spaces or represented spaces, this overloading rarely leads to confusion.
Note that while the names of the concepts are identical, there are significant differences for their behavior in a metric and a represented space context respectively:
A function between metric spaces is continuous if and only if it is sequentially continuous while for represented spaces the backward implication can fail.

\subsection{Recovering continuity on Baire space from a metric structure}\label{sec: baire_metric}

Let $\B$ be a naming space, i.e.\ $\B = \A^\Q$ for some countable, non-empty sets $\Q$ and $\A$.
Since $\Q$ is countable and non-empty there exists a surjective function from $\NN$ to $\Q$.
Such a function may be understood as an enumeration $\mathbf q = (q_n)_{n \in \NN}$ of $\Q$.
For each such enumeration $\mathbf q$ define a mapping $d_{\mathbf q}\colon \B \times \B \to \RR$ by
\[ d_{\mathbf q}(\varphi, \psi) := \begin{cases} 0 & \text{if } \varphi = \psi \\ 2^{-\min\{n\in \NN \mid \varphi(q_n) \neq \psi(q_n)\}} & \text{otherwise.} \end{cases} \]
To argue that this definition is well-formed note that the second case is only assumed if it is not true that $\varphi = \psi$ but in this case there exists some $q$ such that $\varphi(q) \neq \psi(q)$ and since $\mathbf q$ is an enumeration of $\Q$ there exists some $n$ such that $q = q_n$ and thus $\varphi(q_n) \neq \psi(q_n)$.
Thus, the set that the minimum is taken over is not empty and the minimum takes a natural number value.

\begin{prop}[\texttt{dst\_pos}, \texttt{dst\_sym}, \texttt{dstxx}, \texttt{dst\_trngl}, \texttt{dst\_eq}]
  $(\B, d_{\mathbf q})$ is a metric space.
\end{prop}
Let us briefly discuss some details specific to our \Coq{} proof of this proposition.
The discussion before the proposition implies that the definition of $d_{\mathbf q}$ branches over the undecidable proposition whether two elements of a naming space are equal and thus a priori only specifies a relation and not a function.
Now, the axiomatization of the real numbers implies a choice principle that is strong enough to move from this relation to an actual function \cite{boldo2015coquelicot}.
To be able to prove that the thus defined function is a metric it is additionally necessary to assume functional extensionality.
Another tool that is used in the proof to deal with the occurring minima somewhat efficiently is a function named \verb$search$.
There are several adaptations of such a function in other developments some of which could have been reused by assumption of rather mild additional axioms.
The $\verb$search$$ function made an earlier appearance in this paper:
It is what is behind the implementation of the search operator from Example~\ref{ex: continuous search 1}.

Recall that in Section~\ref{sec: continuity} we introduced for any naming space $\B$ the limit operator $\lim_\B$ corresponding to pointwise convergence.
We can now compare this limit operator to the one that corresponds to the metric notion of convergence and was introduced as $\lim_{(M, d)}$ in the introduction of this section.
\begin{thm}[\texttt{lim\_lim}]\label{thm: limlim}
  Let $\B$ be a naming space and $\mathbf q$ an enumeration of its question set.
  Convergence with respect to $d_{\mathbf q}$ is exactly pointwise convergence.
  I.e.\ $\lim_{(\B,d_{\mathbf q})} = \lim_\B$.
\end{thm}
As sequential continuity is defined directly from the notions of convergence, the above theorem implies that the notion of sequential continuity for naming spaces as introduced in Section~\ref{sec: continuity} is also reproduced.
\begin{cor}
  A function between naming spaces is sequentially continuous if and only if it is sequentially continuous as function between the corresponding metric spaces.
\end{cor}
It is possible to generalize the above corollary to apply not only to total but also to partial operators.
We omit the generalized statement here and instead formulate a theorem about the non-sequential version of continuity.
\begin{thm}[\texttt{cont\_cont}]\label{resu: cont_cont}
  Let $\B$ and $\B'$ be naming spaces and $\mathbf q$ and $\mathbf q'$ enumerations of their respective question sets.
  A partial operator $F\colon {\subseteq \B} \to \B'$ is continuous in the sense of Section~\ref{sec: continuity} if and only if it is continuous as a function from $(\dom(F),d_{\mathbf q})$ to $(\B', d_{\mathbf q'})$.
\end{thm}
We omit details of the straightforward proofs of Theorem~\ref{thm: limlim} and Theorem~\ref{resu: cont_cont}.

\begin{exa}[\texttt{examples/continuous\_search.v}]\label{ex: continuous search 3}
  The traditional notion of continuity of partial operators on Baire space $\NN^\NN$ is captured by the continuity introduced in Section \ref{sec: continuity} if all of the questions and answer sets are taken to be the natural numbers.
\end{exa}

\subsection{Comparing continuity in represented and in metric spaces}\label{sec: continuities}

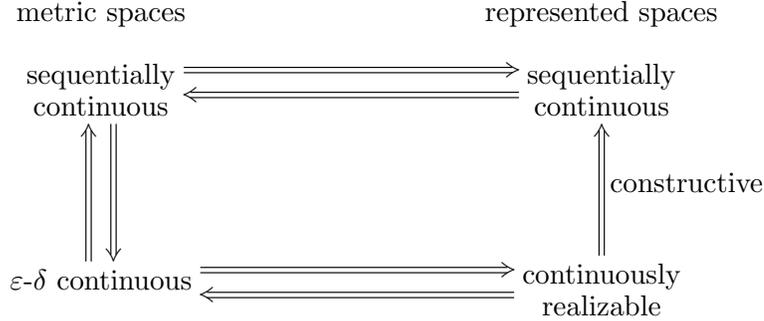
\begin{figure}
  \begin{center}
    \hfil
\xymatrix@!C@C+15ex@R1ex{\text{metric spaces} & \text{represented spaces} \\ \txt{sequentially \\  continuous}\ar@<1ex>@{=>}[ddddd]\ar@<1ex>@{=>}[r]  & \txt{sequentially \\ continuous} \ar@<1ex>@{=>}[l]\\ \\ \\ \\ \\ \varepsilon\text{-}\delta \text{ continuous} \ar@<1ex>@{=>}[r] \ar@<1ex>@{=>}[uuuuu]  & \txt{continuously \\ realizable} \ar@{=>}[uuuuu]_{\txt{constructive}} \ar@<1ex>@{=>}[l]}
\caption{Implications between different notions of continuity on metric spaces.}\label{fig: implications}
\hfil
\end{center}
\end{figure}
A metric space is called \demph{separable} if there exists a dense sequence $(r_n)_{n\in \NN}$ in it.
Here, density of a sequence can equivalently be taken to mean that any open ball in the metric space contains an element of the sequence or that for any point of the metric space there exists a subsequence of $(r_n)$ that converges to that point.
Separable metric spaces are well investigated in computable analysis \cite{weihrauch1993computability}.
To each dense sequence $\mathbf r = (r_n)_{n \in \NN}$ in a metric space $(M, d)$ assign a multifunction $\delta_{\M_{\mathbf r}}\colon \NN^\NN \mto M$ defined by
\[
x \in \delta_{\M_{\mathbf r}}(\varphi) \quad \iff \quad \forall n, d(x, r_{\varphi(n)}) \leq 2^{-n}.
\]
Clearly $\NN^\NN$ is a naming space and $\delta_{\M_{\mathbf r}}$ can be proven singlevalued and surjective (\verb$mrep_sing$ and \verb$mrep_sur$).
Thus $\delta_{\M_{\mathbf r}}$ defines a representation and we denote the corresponding represented space by $\M_{\mathbf r} := (M, \delta_{\M_{\mathbf r}})$.

Note that the construction of $\M_{\mathbf r}$ is very similar to how we chose to represent the real numbers through rational approximations in Example~\ref{ex: q_reals}.
A name of an element of a metric space produces an index of an approximation from an accuracy requirement.
The only difference is that for real numbers we decided to minimize the number of types involved by using rational numbers as both accuracy requirements and approximations while for metric spaces we use accuracy requirements of the form $2^{-n}$ to minimize the number of distinct types by matching the type of indices and setting both question and answer types to be $\NN$.

A sequence $(x_n)$ in a separable metric space $(M, d)$ converges to a limit $x$ from $M$ if and only if the same is true in the represented space $\M_{\mathbf r}$ or in symbols if and only if it holds that $\lim_{\M_{\mathbf r}}(x_n) = x$.
This remains true irrespectively of the choice of the dense sequence that is used for constructing the representation and can more concisely be formulated as $\lim_{(M,d)} = \lim_{\M_{\mathbf r}}$.
Furthermore, if $(M', d')$ is another separable metric space, then a function $f\colon M \to M'$ is continuous as a function between metric spaces if and only if it is continuously realizable as a function $f\colon \M_{\mathbf r} \to \M'_{\mathbf r'}$ and this remains true irrespectively of the choices of the dense sequences $\mathbf r$ and $\mathbf r'$.

This section describes our formal proofs of statements comparing the two continuity notions and their sequential versions (cf. Figure \ref{fig: implications}).
The proofs have been kept as constructive as possible.
Since the definition of a metric space relies on the axiomatic reals, only one of the implications is fully constructive, the others are conservative over the background theory of real numbers and do not rely on the axioms of the real numbers in an essential way.
Let us reiterate that the construction of a represented space from a separable metric space explicitly relies on the choice of a dense sequence.
I.e.\ for each dense sequence $\mathbf r$ in the metric space one obtains a represented space $\M_{\mathbf r}$.
In the following we usually drop the index $\mathbf r$ for brevity.
This is justified in a continuity setting as different choices of dense sequences lead to isomorphic represented spaces.
As always, the situation is more complicated if computability is considered and in this case one should assume in the following that two metric spaces with dense sequences are fixed.
In the formal development, the dense sequences are always present as parameters.

Let us first describe the proof of the equivalence of the notions of sequential continuity.
The main part of the proof is that the notions of limit in the metric space and the corresponding represented space coincide.
\begin{thm}[\texttt{lim\_mlim}]\label{lem:lim-mlim}
  Whenever $(M,d)$ is a separable metric space and $\M$ as above then
  \[ \lim\nolimits_{(M,d)} = \lim\nolimits_\M. \]
\end{thm}
\begin{proof}[idea]
  The proof that convergence in the represented space implies convergence in the metric space is straightforward.
  The idea behind the other direction can be sketched as follows:
  If $(x_n)$ converges to $x$ in the metric space then there exists a modulus of convergence, i.e.\ some $\mu\colon \NN \to \NN$ such that
  \[
  \forall n, m \geq \mu(n)\colon d(x_m, x) \leq 2^{-n}.
  \]
  From an arbitrary sequence $(\varphi'_m)$ of names of $x_m$ and a name $\varphi'$ of $x$ an appropriate convergent sequence of names can be defined by
  \[
  \varphi_m(n) := \begin{cases}
    \varphi'(n+1), & \text{if } \mu(n+1) \leq m \\
    \varphi'_{m}(n), & \text{otherwise}
  \end{cases}
  \]
  and its limit is given by $\varphi(n) := \varphi'(n+1)$ which is clearly a name of $x$ again.
\end{proof}
The availability of a modulus of convergence as a function relies on a use of the axiom of countable choice.
This could probably be eliminated by appropriate assumptions about the values of the metric being approximable on the elements of the dense sequence (i.e.\ by working with computable metric spaces).

That the sequential notions of continuity on metric and represented spaces coincide follows immediately from the above Theorem~\ref{lem:lim-mlim}.
As the proof of each direction requires to translate limits in both directions, either of the directions is as constructive or non-constructive as the worse direction of the previous theorem.
\begin{cor}[\texttt{scnt\_mscnt}]\label{resu: scnt_mscnt}
  Let $(M,d)$ and $(M', d')$ be separable metric spaces and $\M$ and $\M'$ the derived represented spaces.
  Then $f: M \to M'$, is sequentially continuous as a function between metric spaces if and only if it is sequentially continuous as function $f\colon \M \to \M'$.
\end{cor}

Next let us state the equivalence of continuous realizability and $\varepsilon$-$\delta$-continuity.
One implication, namely proving continuous realizability from $\varepsilon$-$\delta$-continuity needs stronger assumptions and for the \Incone{} library we thus separated the proofs.
\begin{lem}[\texttt{cont\_mcont} and \texttt{mcont\_cont}]\label{resu: cont_mcont}
  Let $(M, d)$ and $(M', d')$ be two separable metric spaces.
  A function $f\colon M\to M'$ is $\varepsilon$-$\delta$-continuous if and only if $f\colon \M \to \M'$ is continuous.
\end{lem}
While the proof that continuous realizability implies $\varepsilon$-$\delta$-continuity is straightforward, the proof of the other implication required some work and we sketch some of the details.

Interestingly, the tools needed for this proof are in spirit fairly close to those that were used to prove the existence of associates in Section \ref{sec: universal}, more specifically we also use minimal moduli.
Call a function $m\colon \NN \to \NN$ a \demph{metric modulus of continuity} of $f$ in $x \in M$ if
\[ \forall y\in M\colon d(x, y) \leq 2^{-m(n)} \implies d'(f(x),f(y)) \leq 2^{-n} \]
and call such a modulus \demph{minimal} if it is minimal in the obvious way.
This notion generalizes the one used to prove the existence of associates in Section~\ref{sec: universal} if the naming spaces are equipped with the metric space structures that were introduced in Section~\ref{sec: baire_metric} and reasonable assumptions about the enumeration used for this are made.
\begin{lem}[\texttt{exists\_minmod\_met}]
  For any continuous function $f$ between metric spaces and any argument $x$ for $f$ there exists a minimal modulus of $f$ in $x$.
\end{lem}
This lemma implies its version for naming spaces that we have earlier proven in Section~\ref{sec: universal}.
As we noted in the discussion following Lemma~\ref{lem: exists minmod} this weaker version cannot be proven constructively.
It follows that the same remains true for the lemma above.

The key idea behind the rest of the proof is to also use some notion of being self-modulating.
This is slightly complicated by the fact the notion that made sense for moduli of continuity on naming spaces does not translate to the metric setting.
Moreover, in general the minimal modulus of continuity fails to have some of the nice properties that we relied on for naming spaces.
The function assigning to each $x$ the minimal modulus function of $f$ in $x$ is usually not even continuous.
This is, for instance, not possible if $(M, d)$ is connected as the function takes values in a totally disconnected space and functions from connected to totally disconnected spaces can only be continuous if they are constant.
One might suspect that the awkward typing is the cause for this, and that quantifying the continuity of $f$ in $x$ by use of a function of type $\RR\to \RR$ instead would help, but it does not.
It is known that also in this case returning the minimal modulus of continuity need not be a continuous operation.
A construction of a continuous such assignment, while possible in general, takes considerably more effort \cite{guthrie1983continuous, enayat2000delta}.
Thus, for our proof that $\varepsilon$-$\delta$-continuity implies continuous realizability we rely on a notion of being almost self-modulating instead: The value of the minimal modulus on slightly disturbed input from the metric space can be bounded in terms of a shift of the minimal modulus in the original value.

\subsection{Sierpi\'{n}ski space and closed choice on the naturals}\label{sec: closed choice}

As a somewhat orthogonal class of examples of represented spaces that are far from being metrizable let us consider certain hyper spaces.
A hyper space is a space whose underlying set consists of subsets of a given represented space and depending on the application very different ways to represent such sets of subsets can be appropriate.
An important tool for introducing representations of hyper spaces is Sierpinski space $\Sp$.
The base set of $\Sp$ is the two point set $\{\top, \bot\}$ equipped with the total representation $\delta_\Sp$ with names of type $\NN \to \bool$ specified by
\[ 
\delta_\Sp(\varphi) = \top\quad \iff \quad \exists n \in \NN\ \varphi(n) = \true.
\]
The elements of Sierpinski space should be understood as symbols for termination and divergence.
More specifically: From any kind of computational process for which the notion of an elementary computational step is meaningful one may obtain a sequence of Booleans by taking the $n$-th element to be the truth-value of the statement ``the process finishes within the first $n$ basic steps''.
This way one obtains a name of $\top$ if the process terminates and $\bot$ otherwise.

The final topology of this representation is $\{\emptyset,\{\top\},\Sp\}$ and as a non-Hausdorff space this space is not metrizable.
Without going into details let us mention that also for Sierpinski space there is a tight connection between continuous realizability and topological continuity.
To connect Sierpinski space to hyperspaces consider for a subset $U \subseteq \X$ its characteristic function $\chi_U\colon \X \to \Sp$ defined by
\[
\chi_U(x) := \begin{cases} \top\ \text{ if } x \in U, \\ \bot \ \text{otherwise.}  \end{cases}
\]
For illustration let us assume that we can also safely equip $\X$ with the final topology of the representation without changing the notion of continuity.
Recall that a function is topologically continuous if and only if preimages of open sets are open.
From the topology of Sierpinski space specified above it is evident that in this setting a subset of $\X$ is open if and only if its characteristic function is continuous.
Inspired by this we may generalize and follow sources such as  \cite{pauly2016topological} in identifying the space $\Os(\X)$ of open subsets of $\X$ with the space of continuous functions $\Sp^\X$ as defined in Section \ref{sec: function spaces}.
Similarly, the space $\AA(\X)$ of closed subsets of $\X$ can be introduced as the space of complements of opens.

If $\X$ is taken to be a concrete space it is often possible to give simpler descriptions of $\Os(\X)$ and $\AA(\X)$.
For instance if $\X := \NN$ is the discrete represented space of natural numbers one may make use of a space of infinite sequences as introduced in Section \ref{sec: sequences} and in particular of the special case $I = \NN$ and $\X = \Sp$ of Lemma \ref{resu: sig_iso_fun} which guarantees that $\Os(\NN) = \Sp^\NN \simeq \prod_\NN\Sp = \Sp^\omega$.
There exists a fully concrete description of $\Os(\NN)$ that is often used for reasoning about this space in computable analysis:
The enumeration representation, where a name of an open set enumerates its elements.
We call the corresponding space $\mathcal O_\NN$.
The representation of the corresponding concrete space $\mathcal{A}_\NN$ of closed subsets encodes sets as functions $\varphi\colon \NN \to \NN$ via
\[
\delta_{\mathcal{A}_\NN}(\varphi) := \NN\setminus\{ n \in \NN\mid \exists m\colon \varphi(m) = n+1\}. 
\]
That is: the information that a name $\varphi$ specifies about a closed set is an enumeration of its complement.

We provide a formal proof that the enumeration representations of the open and closed subsets of the natural numbers capture the abstract structure of these spaces through the exponential in the category of represented spaces.
\begin{thm}[\texttt{AN\_iso\_Anat}, \texttt{ON\_iso\_Onat} and \texttt{clsd\_iso\_open}]\label{resu: ON_iso_Onat}
  $\AA(\NN) \simeq \Os(\NN)$ and the concrete spaces of open and closed sets above are correct, i.e.\ $\AA(\NN)\simeq\mathcal A_\NN$ and $\Os(\NN) \simeq \mathcal O_\NN$.
\end{thm}
The first of these isomorphies may look surprising but is evident on closer inspection: the isomorphism is taking the complement and it is realized by the identity function.
The isomorphism of $\Os(\NN)$ and $\mathcal O_\NN$ is proven by first replacing $\Os(\NN)$ by $\Sp^\omega$ as described above.
The realizers for the isomorphisms between $\Sp^\omega$ and $\mathcal O_\NN$ can be defined as functions directly by relying on the Cantor paring function or any standard pairing function.
We chose to use the standard pairing function provided by the mathematical components library.

As an application let us consider choice operators.
To implement a choice operator one has to select an element of a closed set.
As closed sets are represented as complements of open sets and thus by providing positive information about their complement, a realistic implementation of a choice operator is rarely possible.
Such an implementation still solves a problem that can be described by fairly basic means and thus asserting the existence of a solution is popular for classifying unsolvability of other tasks according to their logical strength.
More details about this can be found in survey articles about Weihrauch reducibility such as \cite{brattka2017weihrauch}, more concrete examples of classifications where the problems discussed here make an appearance can be found for instance in \cite{Pauly2018}.

Let us give a precise formulation of what it means to select an element from a closed subset.
For a fixed represented space $\X$ let closed choice on $\X$ be the task of finding a realizer of the multifunction $C_\X$ defined by
\[ C_\X\colon \AA(\X) \mto \X, \quad a \in C_\X(A) \iff a \in A. \]
Or in words: $a$ is an acceptable return value of $C_\X$ on input $A$ if and only if $a$ is an element of $A$.
This in particular means that the domain of $C_\X$ are the non-empty subsets of $\X$.
This is relevant as it means that a potential realizer of $C_\X$ can behave arbitrarily when handed a name of the empty set.
The realizer need not produce any name of an element of $\X$ in this case and may even be undefined.

Consider the special case where $\X =\NN$.
While the argument space of the multivalued function $C_\NN$ is $\AA(\NN)$ we may use the same definition to obtain a multifunction $C_\NN'\colon \mathcal A_\NN \mto \NN$.
A mathematician may even consider it pointless to give this function a new name as $\AA(\NN) \simeq \AA_\NN$ and isomorphic spaces are regularly identified.
Indeed, for the question of whether $C_\NN$ has a continuous realizer the space $\AA(\NN)$ may be substituted with $\mathcal A_\NN$ for this exact reason.
\begin{lem}[\texttt{CN\_CN'\_hcr}]\label{lem:cn-cn'-hcr}
  $C_\NN$ has a continuous realizer if and only if $C_\NN'$ does.
\end{lem}

On the concrete space $\A_\NN$ a standard argument can be used to see that a continuous realizer cannot exist.
\begin{thm}[\texttt{CN'\_not\_cont}]\label{lem:cn'-not-cont}
  $C_\NN'$ does not have a continuous realizer.
\end{thm}
\begin{proof}[outline]
  The proof proceeds by contradiction.
  Assume that to the contrary $C_\NN'$ is continuous and that $F$ is a continuous realizer.
  Pick any name $\varphi$ of the one point set $\{0\}$.
  As $F$ is a realizer, it has to return a name of $0$ on input $\varphi$, i.e.\ $F(\varphi)(\star) = 0$.
  Since $F$ is continuous there is a list $\str q \subseteq \NN$ such that $F(\varphi)(\star) = F(\psi)(\star)$ for all $\psi\colon \NN \to \NN$ that coincide with $\varphi$ on $\str q$.
  Consider the name $\varphi'$ of the non-empty set $A := \NN \setminus \left ( \{n \ | \ \exists m \in \str q, \varphi(m) = n+1 \} \cup \{0\} \right )$ defined by 
  \[
  \varphi'(n) := \begin{cases}
    \varphi(n), & \text{if } n \in \str q \\
    1, & \text{otherwise.}
  \end{cases}
  \]
  On one hand, $F(\varphi')(\star) \in A$ since $F$ is a realizer.
  On the other hand $F(\varphi')(\star) = F(\varphi)(\star) = 0$ as $\varphi$ and $\varphi'$ coincide on $\str q$.
  Now since $0 \notin A$ we arrive at a contradiction and this completes the proof.
\end{proof}

Combining the theorem with the lemma we conclude that closed choice on the natural numbers does not have a continuous realizer:
\begin{cor}[\texttt{CN\_not\_cont}]\label{resu: CN_not_cont}
  Closed choice on the natural numbers does not have a continuous solution.
\end{cor}
Clearly the existence of a continuous solution also rules out the possibility that a computable solution exists.
\section{Conclusion}\label{sec:conclusion}
The \Incone{} library provides general tools for enriching abstract mathematical structures of interest with computational structures and for comparing different such encodings.
It can be used to investigate reasonability of encodings in a fully formal setting, to provide computational content for mathematical statements or to prove it impossible to do so.
We feel that the examples from this paper showcase these capabilities well.
Moreover, they involve many of the most prominent features of \Incone{}:
The abstract definition of the space of open subsets is based on the library's function space construction and our proof of isomorphy relies on the fact that in this specific case the corresponding space of functions is isomorphic to a space of sequences.
We believe the \Incone{} library to be reasonably accessible for the computable analysis community and hope that its combination with methods from that community \cite{DBLP:journals/corr/MullerPP016} could help to make parts of it more accessible to the numerical analysis community.

Let us make an effort to outline some relations to other approaches.
For a given representation, i.e.\ singlevalued and co-total (surjective) multifunction $\delta\colon \B \mto X$, one may define a partial equivalence relation on Baire space via
\[ \varphi \sim \psi \iff \exists x\in X\colon x \in \delta(\varphi) \land x \in \delta(\psi). \]
If $\delta$ is considered a partial function the right-hand side can be replaced by $\delta(\varphi) = \delta(\psi)$ if the equality is interpreted in the appropriate way.
The resulting structure $(\B, \sim)$ is what is called a setoid in constructive mathematics.
Conversely, for each such setoid one may take $X$ to be the set of equivalence classes with respect to $\sim$.
This makes the quotient mapping a representation and one may check that the constructions are ``inverse'' in an appropriate sense.
Other notions for representations can be translated accordingly.
For instance, a function that preserves equivalence induces a function on the equivalence classes and realizes this function with respect to the quotient mappings as representations.
If the function does not preserve the equivalence relation, it induces a multifunction on the equivalence classes that it still realizes.

While mathematicians take quotients without thinking twice, in a type theoretic setting quotients are a problematic operation.
Our proof assistant of choice, \Coq{}, does not support quotient types and there are reasons why it refrains from doing so.
As setoids allow to still reason about the quotient without really taking it, they are a very popular tool in the \Coq{} community and constructive mathematics in general.
The specification of an appropriate type $X$ and a representation that acts as quotient mapping is additional information that has to be provided by the user.
This extra information is exclusively used for specification purposes and we aim to give the greatest possible amount of freedom in how the user wants to present this information.
For instance our choice to treat representations as relations and not as functions can be justified by this.
This comes with the drawback that definability of a function on the description of the quotient need no longer be related to it being realistically implementable.
For instance consider Sierpinski space:
In this case the description of the quotient is a discrete two element set which means that all self-maps are definable.
However, the one that switches $\top$ and $\bot$ is not implementable on the level of names.
As the information from the description of the quotient is computationally irrelevant, another way of introducing it is via an axiomatization similar to that of the real numbers from \Coq{}'s standard library that was used in this paper.
Ideally such an axiomatization should avoid stating unrealistic assumptions as computationally relevant facts but the example of the reals shows that even if it does, this can be worked around.

The idea behind the constructions provided by \Incone{} is to support a user in conveniently obtaining descriptions of quotient types by piecing them together from basic building blocks.
As basic spaces \Incone{} defines discrete spaces, separable metric spaces with a distinguished dense sequence and Sierpinski space.
One may combine these basic spaces via building products, sums, infinite sequences and spaces of functions.
The space $\Os(\NN)$ is an example for such a combination as it is constructed as a space of functions between elementary defined spaces.
This particular space, sometimes also called $\mathcal{P}(\omega)$, is a common sight in work on constructive or computational analysis.
This is even though as a space it is very different from the objects that are actually of interest in these fields.
However, many spaces of practical importance can be identified with a subspace of this space:
Any quasi-Polish space in the sense of de Brecht and therefore in particular any Polish space embeds isomorphically into $\Os(\NN)$ as a subspace of low descriptive complexity \cite{DEBRECHT2013356}.
This covers all of the spaces considered in Section~\ref{sec: metric spaces}.
Let us give a concrete example by outlining of how to isomorphically embed the real numbers into $\Os(\NN) \simeq \Sp^\NN$.
First enumerate the open rational intervals and then send a real number $x$ to the function that on input of a natural number $n$ returns $\top$ if $x$ is contained in the $n$-th interval and $\bot$ if this is not the case.
As checking inclusion of a real number into an open interval is semi-decidable relative to availability of approximations to the real number, we obtain an element of $\Os(\NN)$ that continuously depends on the real number.
This can be checked to lead to a continuous, and even computable, embedding of the reals into $\Os(\NN)$.
Of course, this is nothing but the identification of a real number with a Dedekind cut.

For future directions in the development of \Incone{} we plan to work on additional applications.
One particularly fitting extension of the contents of this paper would be a proof that $C([0,1]) \simeq \RR^{[0,1]}$.
This statement is called the Computable Weierstra\ss\ Theorem \cite{pour1975simple}:
Here, $C([0,1])$ are the continuous real-valued functions on the unit interval represented as separable metric space with supremum norm and the rational polynomials as dense sequence.
The function-space $\RR^{[0,1]}$ has the same underlying set but uses a different representation that the theorem states to be computably equivalent.
Other possibilities include:
\begin{itemize}
\item Results about ODE solving \cite{ImmlerH12,MakarovS13, kawamura2018parameterized}.
  This may be done by providing an interface with \corn, parts of it could also be done separately by relying on libraries such as Coq-Interval \cite{melquiond2008proving}.
\item Duality theory for spaces of summable sequences ($\ell_p$-spaces) which provide a pool of examples where sub-spaces of exponentials can be treated complexity theoretically \cite{article,DBLP:conf/lics/SchroderS17}.
  Additionally it constitutes a step towards capturing popular methods for solving partial differential equations \cite{BRATTKA2006858,Selivanova:2008:CSS:1480533.1480596,BCF17a}.
\item A characterization of continuity via preimages of open sets, general considerations about admissibility, discreteness, compactness and many other similar results \cite{pauly2016topological,schroeder_phd}.
\end{itemize}
There are also still some gaps in the \Incone{} library that we plan to fill.
Currently there is no complete proof that the category of represented spaces is Cartesian closed as the corresponding universal property is missing.
The library proves a restricted case by providing a duality operator, but a full proof would be desirable.
Further steps would be to use a formalization of a model of computation to make it possible to talk about computability without moving to the meta level and using the Prop/Type distinction.
This would have the additional advantage of opening the possibility to do complexity theory.
Even without reliance on a model of computation it should be possible to add the capability to do qualitative complexity theory in terms of tracking the decrease in accuracy of approximations \cite{kawamura_complexity_2012,Fre:2017:GSA:3082405.3082461,DBLP:journals/corr/abs-1711-10530,tamc2019}.
As the \Incone{} library keeps close to recent work about complexity theory for computable analysis, this should be possible with relatively low effort.
A full treatment of step-counting complexity might become available in the not too distant future due to recent progress on the formalization of models of computation \cite{ForsterSmolka:2018:Computability-JAR,MaximilianWuttke} and methods from implicit complexity theory \cite{feree2018formal}.
Another way to gain insight into such efficiency considerations would be to capture the trace of the basic feasible functionals on the operators on Baire space \cite{MEHLHORN1976147, Kapron1996ANC, Kapron2018TypetwoPA}.

\bibliography{cc}{}

\newcommand{\etalchar}[1]{$^{#1}$}
\begin{thebibliography}{CFGW04}

\bibitem[AB14]{avigad_brattka_2014}
Jeremy Avigad and Vasco Brattka.
\newblock {\em Computability and analysis: the legacy of {A}lan {T}uring}, pages
  1--47.
\newblock Lecture Notes in Logic. Cambridge University Press, 2014.

\bibitem[ACK{\etalchar{+}}20]{affeldt:hal-02463336}
Reynald Affeldt, Cyril Cohen, Marie Kerjean, Assia Mahboubi, Damien Rouhling,
  and Kazuhiko Sakaguchi.
\newblock {Formalizing functional analysis structures in dependent type
  theory}.
\newblock working paper or preprint, January 2020.

\bibitem[ACR18]{JFR8124}
Reynald Affeldt, Cyril Cohen, and Damien Rouhling.
\newblock Formalization techniques for asymptotic reasoning in classical
  analysis.
\newblock {\em Journal of Formalized Reasoning}, 11(1):43--76, 2018.

\bibitem[ASBZ13]{10.1007/978-3-642-39053-1_1}
Klaus Ambos-Spies, Ulrike Brandt, and Martin Ziegler.
\newblock Real benefit of promises and advice.
\newblock In Paola Bonizzoni, Vasco Brattka, and Benedikt L{\"o}we, editors,
  {\em The Nature of Computation. Logic, Algorithms, Applications}, pages
  1--11, Berlin, Heidelberg, 2013. Springer Berlin Heidelberg.

\bibitem[Bau00]{Bauer:2000:RAC:933370}
Andrej Bauer.
\newblock {\em The Realizability Approach to Computable Analysis and Topology}.
\newblock PhD thesis, Carnegie Mellon University, Pittsburgh, PA, USA, 2000.
\newblock AAI3002721.

\bibitem[BB12]{bishop2012constructive}
Errett Bishop and Douglas Bridges.
\newblock {\em Constructive analysis}, volume 279.
\newblock Springer Science \& Business Media, 2012.

\bibitem[BCC{\etalchar{+}}06]{balluchi2006ariadne}
Andrea Balluchi, Alberto Casagrande, Pieter Collins, Alberto Ferrari, Tiziano
  Villa, and Alberto~L Sangiovanni-Vincentelli.
\newblock Ariadne: a framework for reachability analysis of hybrid automata.
\newblock In {\em In: Proceedings of the International Syposium on Mathematical
  Theory of Networks and Systems.(2006}. Citeseer, 2006.

\bibitem[BCF{\etalchar{+}}17]{BCF17a}
Sylvie Boldo, Fran{\c c}ois Cl{\'e}ment, Florian Faissole, Vincent Martin, and
  Micaela Mayero.
\newblock {A {C}oq Formal Proof of the {L}ax--{M}ilgram theorem}.
\newblock In {\em {Proceedings of the 6th ACM SIGPLAN Conference on Certified
  Programs and Proofs}}, CPP 2017, pages 79--89, Paris, France, January 2017.
  ACM.

\bibitem[BDBP12]{brattka2012closed}
Vasco Brattka, Matthew De~Brecht, and Arno Pauly.
\newblock Closed choice and a uniform low basis theorem.
\newblock {\em Annals of Pure and Applied Logic}, 163(8):986--1008, 2012.

\bibitem[BG11]{brattka2011effective}
Vasco Brattka and Guido Gherardi.
\newblock Effective choice and boundedness principles in computable analysis.
\newblock {\em Bulletin of Symbolic Logic}, 17(1):73--117, 2011.

\bibitem[BGP17]{brattka2017weihrauch}
Vasco Brattka, Guido Gherardi, and Arno Pauly.
\newblock Weihrauch complexity in computable analysis.
\newblock {\em arXiv preprint arXiv:1707.03202}, 2017.

\bibitem[BKMP16]{brattka_et_al:DR:2016:5686}
Vasco Brattka, Akitoshi Kawamura, Alberto Marcone, and Arno Pauly.
\newblock {Measuring the Complexity of Computational Content (Dagstuhl Seminar
  15392)}.
\newblock {\em Dagstuhl Reports}, 5(9):77--104, 2016.

\bibitem[BLM15]{boldo2015coquelicot}
Sylvie Boldo, Catherine Lelay, and Guillaume Melquiond.
\newblock Coquelicot: A user-friendly library of real analysis for {C}oq.
\newblock {\em Mathematics in Computer Science}, 9(1):41--62, 2015.

\bibitem[BLM16]{BLM16}
Sylvie Boldo, Catherine Lelay, and Guillaume Melquiond.
\newblock Formalization of real analysis: {A} survey of proof assistants and
  libraries.
\newblock {\em Mathematical Structures in Computer Science}, 26(7):1196--1233,
  2016.

\bibitem[BM17]{BolMel17}
Sylvie Boldo and Guillaume Melquiond.
\newblock {\em Computer Arithmetic and Formal Proofs}.
\newblock ISTE Press -- Elsevier, 2017.

\bibitem[Bro75]{MR0532661}
L.~E.~J. Brouwer.
\newblock {\em Collected works. {V}ol. 1}.
\newblock North-Holland Publishing Co., Amsterdam-Oxford; American Elsevier
  Publishing Co., Inc., New York, 1975.
\newblock Philosophy and foundations of mathematics, Edited by A. Heyting.

\bibitem[BY06]{BRATTKA2006858}
Vasco Brattka and Atsushi Yoshikawa.
\newblock Towards computability of elliptic boundary value problems in
  variational formulation.
\newblock {\em Journal of Complexity}, 22(6):858 -- 880, 2006.
\newblock Computability and Complexity in Analysis.

\bibitem[CFGW04]{cruz2004c}
Lu{\'\i}s Cruz-Filipe, Herman Geuvers, and Freek Wiedijk.
\newblock {C-CoRN}, the constructive {C}oq repository at {N}ijmegen.
\newblock In {\em International Conference on Mathematical Knowledge
  Management}, pages 88--103. Springer, 2004.

\bibitem[dB13]{DEBRECHT2013356}
Matthew de~Brecht.
\newblock Quasi-polish spaces.
\newblock {\em Annals of Pure and Applied Logic}, 164(3):356 -- 381, 2013.

\bibitem[Dei92]{deimling1992multivalued}
K.~Deimling.
\newblock {\em Multivalued Differential Equations}.
\newblock De Gruyter series in nonlinear analysis and applications. W. de
  Gruyter, 1992.

\bibitem[EM46]{10.2307/2371832}
Samuel Eilenberg and Deane Montgomery.
\newblock Fixed point theorems for multi-valued transformations.
\newblock {\em American Journal of Mathematics}, 68(2):214--222, 1946.

\bibitem[Ena00]{enayat2000delta}
Ali Enayat.
\newblock $\delta$ as a continuous function of x and $\varepsilon$.
\newblock {\em The American Mathematical Monthly}, 107(2):151--155, 2000.

\bibitem[EX16]{ESCARDO2016770}
Mart{\'\i}n Escard{\'o} and Chuangjie Xu.
\newblock A constructive manifestation of the {K}leene-{K}reisel continuous
  functionals.
\newblock {\em Annals of Pure and Applied Logic}, 167(9):770 -- 793, 2016.
\newblock Fourth Workshop on Formal Topology (4WFTop).

\bibitem[F{\'e}r17]{Fre:2017:GSA:3082405.3082461}
Hugo F{\'e}r{\'e}e.
\newblock Game semantics approach to higher-order complexity.
\newblock {\em J. Comput. Syst. Sci.}, 87(C):1--15, August 2017.

\bibitem[FHM{\etalchar{+}}18]{feree2018formal}
Hugo F{\'e}r{\'e}e, Samuel Hym, Micaela Mayero, Jean-Yves Moyen, and David
  Nowak.
\newblock Formal proof of polynomial-time complexity with
  quasi-interpretations.
\newblock In {\em Proceedings of the 7th ACM SIGPLAN International Conference
  on Certified Programs and Proofs}, pages 146--157. ACM, 2018.

\bibitem[FS18]{ForsterSmolka:2018:Computability-JAR}
Yannick Forster and Gert Smolka.
\newblock {C}all-by-{V}alue {L}ambda {C}alculus as a {M}odel of {C}omputation
  in {C}oq.
\newblock {\em Journal of Automated Reasoning}, 2018.

\bibitem[GMT16]{gonthier:inria-00258384}
Georges Gonthier, Assia Mahboubi, and Enrico Tassi.
\newblock {A Small Scale Reflection Extension for the Coq system}.
\newblock Research Report RR-6455, {Inria Saclay Ile de France}, 2016.

\bibitem[Grz57]{MR0089809}
A.~Grzegorczyk.
\newblock On the definitions of computable real continuous functions.
\newblock {\em Fund. Math.}, 44:61--71, 1957.

\bibitem[Gut83]{guthrie1983continuous}
Joe~A Guthrie.
\newblock A continuous modulus of continuity.
\newblock {\em The American Mathematical Monthly}, 90(2):126--127, 1983.

\bibitem[IH12]{ImmlerH12}
Fabian Immler and Johannes H{\"{o}}lzl.
\newblock Numerical analysis of ordinary differential equations in
  {I}sabelle/{HOL}.
\newblock In {\em ITP}, volume 7406 of {\em LNCS}, pages 377--392, 2012.

\bibitem[KC96]{Kapron1996ANC}
Bruce~M. Kapron and Stephen~A. Cook.
\newblock A new characterization of type-2 feasibility.
\newblock {\em SIAM J. Comput.}, 25:117--132, 1996.

\bibitem[KC12]{kawamura_complexity_2012}
Akitoshi Kawamura and Stephen Cook.
\newblock Complexity theory for operators in analysis.
\newblock {\em ACM Transactions in Computation Theory}, 4(2):Article 5, 2012.

\bibitem[Kle59]{kleene1959constructivity}
S.C. Kleene.
\newblock Countable functionals.
\newblock {\em Constructivity in Mathematics: proceedings of the colloquium
  held at Amsterdam}, 1959.

\bibitem[Ko91]{MR1137517}
Ker-I Ko.
\newblock {\em Complexity theory of real functions}.
\newblock Progress in Theoretical Computer Science. Birkh{\"a}user Boston Inc.,
  Boston, MA, 1991.

\bibitem[Kon08]{konecny2008aern}
Michal Konecn{\`y}.
\newblock {AERN-Real: Arbitrary-precision interval arithmetic for approximating
  exact real numbers}, 2008.

\bibitem[Kre59]{kreisel1959interpretation}
Georg Kreisel.
\newblock Interpretation of analysis by means of constructive functionals of
  finite type, {C}onstructivity in {M}athematics, 1959.

\bibitem[KS18]{Kapron2018TypetwoPA}
Bruce~M. Kapron and Florian Steinberg.
\newblock Type-two polynomial-time and restricted lookahead.
\newblock In {\em LICS}, 2018.

\bibitem[KST18]{kawamura2018parameterized}
Akitoshi Kawamura, Florian Steinberg, and Holger Thies.
\newblock Parameterized complexity for uniform operators on multidimensional
  analytic functions and {ODE} solving.
\newblock In {\em International Workshop on Logic, Language, Information, and
  Computation}, pages 223--236. Springer, 2018.

\bibitem[KST19]{tamc2019}
Akitoshi Kawamura, Florian Steinberg, and Holger Thies.
\newblock Second-order linear-time computability with applications in
  computable analysis.
\newblock {\em 15th Annual Conference on Theory and Applications of Models of
  Computation}, 2019.
\newblock extended abstract accepted for presentation at TAMC 2019.

\bibitem[KW85]{kreitz1985theory}
Christoph Kreitz and Klaus Weihrauch.
\newblock Theory of representations.
\newblock {\em Theoretical computer science}, 38:35--53, 1985.

\bibitem[Lac58]{MR0105357}
Daniel Lacombe.
\newblock Sur les possibilit\'es d'extension de la notion de fonction
  r\'ecursive aux fonctions d'une ou plusieurs variables r\'eelles.
\newblock In {\em Le raisonnement en math\'ematiques et en sciences
  exp\'erimentales}, Colloques Internationaux du Centre National de la
  Recherche Scientifique, LXX, pages 67--75. Editions du Centre National de la
  Recherche Scientifique, Paris, 1958.

\bibitem[LN15]{longley2015higher}
John Longley and Dag Normann.
\newblock {\em Higher-order computability}, volume 100.
\newblock Springer, 2015.

\bibitem[Meh76]{MEHLHORN1976147}
Kurt Mehlhorn.
\newblock Polynomial and abstract subrecursive classes.
\newblock {\em Journal of Computer and System Sciences}, 12(2):147 -- 178,
  1976.

\bibitem[Mel08]{melquiond2008proving}
Guillaume Melquiond.
\newblock Proving bounds on real-valued functions with computations.
\newblock In {\em International Joint Conference on Automated Reasoning}, pages
  2--17. Springer, 2008.

\bibitem[MPPZ16]{DBLP:journals/corr/MullerPP016}
Norbert~Th. M{\"{u}}ller, Sewon Park, Norbert Preining, and Martin Ziegler.
\newblock On formal verification in imperative multivalued programming over
  continuous data types.
\newblock {\em CoRR}, abs/1608.05787, 2016.

\bibitem[MS13]{MakarovS13}
Evgeny Makarov and Bas Spitters.
\newblock The {P}icard algorithm for ordinary differential equations in {C}oq.
\newblock In {\em {ITP}}, volume 7998 of {\em Lecture Notes in Computer
  Science}, pages 463--468. Springer, 2013.

\bibitem[M{\"{u}}l01]{zbMATH01746043}
Norbert~Th. M{\"{u}}ller.
\newblock {The {iRRAM}: Exact arithmetic in {C++}.}
\newblock In {\em {Computability and complexity in analysis. 4th international
  workshop, CCA 2000. Swansea, GB, September 17--19, 2000. Selected papers}},
  pages 222--252. Berlin: Springer, 2001.

\bibitem[NS20]{DBLP:journals/corr/abs-1711-10530}
Eike Neumann and Florian Steinberg.
\newblock Parametrised second-order complexity theory with applications to the
  study of interval computation.
\newblock {\em Theoret. Comput. Sci.}, 806:281--304, 2020.

\bibitem[O'C05]{10.1007/11541868_16}
Russell O'Connor.
\newblock Essential incompleteness of arithmetic verified by {C}oq.
\newblock In Joe Hurd and Tom Melham, editors, {\em Theorem Proving in Higher
  Order Logics}, pages 245--260, Berlin, Heidelberg, 2005. Springer Berlin
  Heidelberg.

\bibitem[O'C09]{conorPhD}
Russell O'Connor.
\newblock {\em Incompleteness \& Completeness: Formalizing Logic and Analysis
  in Type Theory}.
\newblock PhD thesis, Radboud Universiteit Nijmegen, 2009.

\bibitem[Pau16]{pauly2016topological}
Arno Pauly.
\newblock On the topological aspects of the theory of represented spaces.
\newblock {\em Computability}, 5(2):159--180, 2016.

\bibitem[Pau17]{pauly2012multi}
Arno Pauly.
\newblock Many-one reductions and the category of multivalued functions.
\newblock volume~27, pages 376--404. Cambridge University Press, 2017.

\bibitem[PEC75]{pour1975simple}
Marian~Boykan Pour-El and Jerome Caldwell.
\newblock On a simple definition of computable function of a real variable-with
  applications to functions of a complex variable.
\newblock {\em Mathematical Logic Quarterly}, 21(1):1--19, 1975.

\bibitem[PER89]{pour-el}
Marian~B. Pour-El and J.~Ian Richards.
\newblock {\em Computability in Analysis and Physics}, volume Volume 1 of {\em
  Perspectives in Mathematical Logic}.
\newblock Springer-Verlag, Berlin, 1989.

\bibitem[PS18]{Pauly2018}
Arno Pauly and Florian Steinberg.
\newblock Comparing representations for function spaces in computable analysis.
\newblock {\em Theory of Computing Systems}, 62(3):557--582, Apr 2018.

\bibitem[PZ13]{Pauly2013RelativeCA}
Arno Pauly and Martin Ziegler.
\newblock Relative computability and uniform continuity of relations.
\newblock {\em J. Logic \& Analysis}, 5, 2013.

\bibitem[Sch02a]{schroeder_phd}
Matthias Schr{\"o}der.
\newblock {\em Admissible Representations for Continuous Computations}.
\newblock PhD thesis, FernUniversit{\"a}t Hagen, 2002.

\bibitem[Sch02b]{schroder2002extended}
Matthias Schr{\"o}der.
\newblock Extended admissibility.
\newblock {\em Theoretical computer science}, 284(2):519--538, 2002.

\bibitem[Sch04]{article}
Matthias Schr\"{o}der.
\newblock Spaces allowing type-2 complexity theory revisited.
\newblock {\em Math. Log. Q.}, 50:443--459, 09 2004.

\bibitem[Sel94]{SELMAN1994357}
Alan~L. Selman.
\newblock A taxonomy of complexity classes of functions.
\newblock {\em Journal of Computer and System Sciences}, 48(2):357 -- 381,
  1994.

\bibitem[Soa78]{soare1978recursively}
Robert~I Soare.
\newblock Recursively enumerable sets and degrees.
\newblock {\em Bulletin of the American Mathematical Society},
  84(6):1149--1181, 1978.

\bibitem[SS08]{Selivanova:2008:CSS:1480533.1480596}
Svetlana Selivanova and Victor Selivanov.
\newblock Computing solutions of symmetric hyperbolic systems of pde's.
\newblock {\em Electron. Notes Theor. Comput. Sci.}, 221:243--255, December
  2008.

\bibitem[SS17]{DBLP:conf/lics/SchroderS17}
Matthias Schr{\"{o}}der and Florian Steinberg.
\newblock Bounded time computation on metric spaces and {B}anach spaces.
\newblock In {\em 32nd Annual {ACM/IEEE} Symposium on Logic in Computer
  Science, {LICS} 2017, Reykjavik, Iceland, June 20-23, 2017}, pages 1--12.
  {IEEE} Computer Society, 2017.

\bibitem[Ste19a]{incone}
Florian Steinberg.
\newblock The \textsc{Incone}{} library.
\newblock \url{https://github.com/FlorianSteinberg/incone}, 2019.
\newblock release v1.0.

\bibitem[Ste19b]{metric}
Florian Steinberg.
\newblock The \textsc{Metric}{} library.
\newblock \burl{https://github.com/FlorianSteinberg/metric}, 2019.
\newblock release v1.0.

\bibitem[Ste19c]{mf}
Florian Steinberg.
\newblock The \textsc{Mf}{} library.
\newblock \burl{https://github.com/FlorianSteinberg/mf}, 2019.
\newblock release v1.0.

\bibitem[Ste19d]{rlzrs}
Florian Steinberg.
\newblock The \textsc{Rlzrs}{} library.
\newblock \burl{https://github.com/FlorianSteinberg/rlzrs}, 2019.
\newblock release v1.0.

\bibitem[Tur36]{turing_computable_1936}
A.~M. Turing.
\newblock On computable numbers, with an application to the
  {Entscheidungsproblem}.
\newblock {\em Proceedings of the London Mathematical Society}, 2(1):230--265,
  1936.

\bibitem[Tur38]{turing1938computable}
Alan~Mathison Turing.
\newblock On computable numbers, with an application to the
  {E}ntscheidungsproblem. a correction.
\newblock {\em Proceedings of the London Mathematical Society}, 2(1):544--546,
  1938.

\bibitem[TvD88]{MR966421}
A.~S. Troelstra and D.~van Dalen.
\newblock {\em Constructivism in mathematics. {V}ol. {II}}, volume 123 of {\em
  Studies in Logic and the Foundations of Mathematics}.
\newblock North-Holland Publishing Co., Amsterdam, 1988.

\bibitem[vO11]{vanoosten2011}
Jaap van Oosten.
\newblock Partial combinatory algebras of functions.
\newblock {\em Notre Dame J. Formal Logic}, 52(4):431--448, 07 2011.

\bibitem[Wei93]{weihrauch1993computability}
Klaus Weihrauch.
\newblock Computability on computable metric spaces.
\newblock {\em Theoretical Computer Science}, 113(2):191--210, 1993.

\bibitem[Wei00]{weihrauch_computable_2000}
Klaus Weihrauch.
\newblock {\em Computable {Analysis}}.
\newblock Springer, Berlin/Heidelberg, 2000.

\bibitem[Wut18]{MaximilianWuttke}
Maximilian Wuttke.
\newblock Verified programming of {T}uring {M}achines in {C}oq.
\newblock Master's thesis, Saarland University, 2018.

\end{thebibliography}
\appendix

\section{An overview of the \Incone{} library}\label{app: incone}
The goal of the \Incone{} library is to provide a \Coq{} formalization of some of the most important concepts and basic results from computable analysis.
The library uses notations that allow to state theorems similar to how they would be formulated in natural language and therefore should be readable for a mathematician familiar with computable analysis without having a deep knowledge of \Coq{}.
For example, Lemma~\ref{lem: tight comp} corresponds to the following statement in the library.
\begin{verbatim}
  Lemma tight_comp R S T (f f': T ->> R) (g g': S ->> T):
  f \tightens f' -> g \tightens g' -> (f \o g) \tightens (f' \o g').
\end{verbatim}
Here, \verb$tight_comp$ is the name of the lemma and the symbols following this name until the colon at the end of the line are parameters of its statement.
These parameters can alternatively be thought of as universally quantified.
The phrase \verb$(f f': T ->> R)$ thus translates to the mathematical language used in the paper as ``for all $f, f'\colon T \mto R$'', or ``let $f$ and $f'$ be multifunctions from $T$ to $R$''.
That $R$ and $T$ are sets can be omitted as it can be inferred from their use.
The notation for multifunctions may seem unfortunate at first but the closer match of a notation \verb$=>$ seemed unwise to use as it is commonly overloaded with a number of different meanings.
The colon is followed by the body of the lemma, where \verb$->$ is the logical implication and \verb$_ \o _$ and \verb$_ \tightens _$ are the notations for composition and the tightening order on multivalued functions.
Many of the infix notations in the library start in a \verb$\$ to avoid blocking too many keywords.

Most statements in this paper and many additional results can be found in the files of the \Incone{} library.
In the paper, we labeled each theorem and lemma with the names of the relevant statements in the library.
This should make it easy to locate and inspect them in \Incone{}.
The structure of this paper and the formal development of the concepts are quite similar.
For most parts making the connection is about as complicated as indicated by the example above.
The most important exception where one needs additional information for the translation is that the library uses the expression ``continuity space'' for what is usually referred to as ``represented space''.
We decided to use this notion since represented spaces are tied to computability theory, which the \Incone{} library avoids except on the meta-level.
This has another important consequence for the translation between the results of the paper and the formal development: The formal versions will only prove continuity where the paper claims computability.
The user can verify the computability claim by inspecting the definitions of the realizers.
Whenever a result proves a function with name \verb$fun_name$ continuous, there will be a function whose name contains the string \verb$fun_name_rlzr$.
It can be checked that this definition does not rely on any assumptions, so that it certifies computability.
How to check this manually is described in more detail below.

The smaller differences between the formal and the informal development of concepts include the definition of the space of functions $\Y^\X$:
While in this paper we chose the underlying set to be the continuous functions, in the formal development we use the co-domain of the function-space representation for the definition.
In \Incone{} the space of functions is denoted by \texttt{cs\_fun} with the notation \verb$_ c-> _$ and that its underlying set contains exactly the continuous functions is proven separately (\verb$ass_cont$).
The reason for this choice is to minimize the strength of the axioms that are automatically assumed whenever spaces of functions are used.
The proof that every continuous function has an associate requires a considerably stronger background theory than most of the other basic results one would be interested in.
Another design choice for the sake of minimizing the use of axioms is that \Incone{} does not use the \Coq-internal notion of equality for equality of multivalued functions.
Instead, it defines this equality as an equivalence relation \verb$_ =~= _$.
This avoids the need to assume the axiom of propositional extensionality that, while being theoretically justified and fairly commonly used, is known for having caused inconsistencies in the past.
As such equivalence relations are a common sight in constructive mathematics, \Coq{} provides the necessary infrastructure to make most operations equally effortless for equalities and equivalence relations.
Thus this detail can mostly be ignored by the user.

The \Incone{} library has some sublibraries that are of separate interest and have fewer dependencies.
These parts can be obtained individually and we give an overview over their purpose in Section \ref{sec: sublibs}.
Before we do so, let us sketch what steps need to be taken to get an installation of \Incone{} up and running and verify the correctness of the results described in this paper.
We will assume that the user has successfully installed \Coq{} on his machine, instructions of how to do so are available online.

\subsection{How to get started}
The paper describes release version $1.0$ of \Incone{}.
As \Incone{} is an ongoing project, newer versions will become available and for a description of the newest version, the GitHub repository can be consulted \cite{incone}.
For readers only interested in verifying the results of the paper, we give instructions how to install and use version $1.0$ of the library here.
This version of \Incone{} has been tested with version 8.9.0 of the \Coq{} proof assistant.
Additionally, the Coquelicot library (version 3.0) and the ssreflect package from math-comp (version 1.7.0) needs to be installed.
To install \Incone{}, the libraries \Mf{}, \Rlzrs{}, \Metric{} and \Incone{} have to be downloaded and installed in this order.
They can be found in their respective GitHub repositories \cite{mf,metric,rlzrs,incone}.
Alternatively, the paper's project page \url{https://holgerthies.github.io/continuity} offers a package containing all libraries.

All of the claims the paper makes about the incone library can be checked by installing the library, opening a new
\Coq{} file with the following includes in the preamble.
\begin{verbatim}
From mf Require Import choice_mf.
From rlzrs Require Import all_rlzrs.
From metric Require Import all_metric Qmetric.
From incone Require Import all_cs classical_func 
                           classical_cont classical_mach
                           Duop Q_reals baire_metric.
\end{verbatim}

Each result in the paper states the name of the corresponding lemma in \Incone{} in brackets after the number of the result.
The formal lemmas and theorems can be inspected via \verb$Check lem_name$, definitions can be detailed using \verb$Print def_name$ and notations by \verb$Locate "not_name"$.
In the case where multiple results share a name, \verb$Locate "lem_name"$ lists all lemmas with that name and unique identifiers and \verb$Print Assumptions lem_name$ shows the axioms that the result assumes.
To list all results involving a concept \Coq’s search function can be used via \verb$Search _ (concept)$ and \verb$Search "phrase"$ might be useful where \verb$phrase$ is an expected substring of a lemma's name.
\Coq{}'s search function can also be used with multiple arguments and the underscore above is due to the first argument being treated differently.
Omitting it will lead to only functions returning \verb$concept$ being listed as a result and a warning that this may not be the intended behavior.

For example, one of the theorems from the paper has the name \verb$cont_comp$.
Typing
\begin{verbatim}
Check cont_comp.
\end{verbatim}
prints the following statement:
\begin{verbatim}
cont_comp
     : forall (X Y Z : cs) (f : Y -> Z) (g : X -> Y), 
        f \is_continuous -> g \is_continuous -> (f \o_f g) \is_continuous
\end{verbatim}
Here, \verb$cs$ is a short notation for \verb$continuity_space$, which is the \Incone{} equivalent of a represented space.
The notation \verb$_ \o_f _$ is for composition of functions, the shorter \verb$_ \o _$ is preserved for composition of multivalued functions.
The rest of the notations should be self-explanatory but it may be worth noting that the use of a $\backslash$ for non-prefix notations is to avoid blocking too many keywords.
The user may first use \verb$Locate "_ \is_continuous"$ to find that this is a notation for the definition \verb$continuous$.
Then the details of the definition can be printed using the command \verb$Print continuous$.
The above may or may not be the desired theorem from the paper.
Typing
\begin{verbatim}
Locate cont_comp.
\end{verbatim}
shows that there are several lemmas with this name and the location of the files that contain them.

For users that are not only interested in verifying the results from the paper, the examples folder is a good starting point.
Most of the the examples are documented via comments in the files and contain some rudimentary proofs that can be used to get familiar with using the library.
A reader familiar with \Coq{} will notice right away that these proofs heavily use the ssreflect proof language.
This is true for the whole of the development and a comprehensible overview over this language and its advantages can be found in \cite{gonthier:inria-00258384}.
\subsection{The structure of \Incone{} and its sub-libraries}\label{sec: sublibs}
Some concepts described in this paper are expected to have applications that are unrelated to computable analysis and the most important ones have been exported to small libraries that can be installed separately and have less dependencies.
For instance, due to its potential applications outside of \Incone{} the development of a convenient environment for manipulation of multifunctions was exported and can be obtained separately as the \Mf{}-library \cite{mf}.
Already the \Incone{} library uses multifunctions for several different purposes:
Through the \Rlzrs{} library for the formulation of realizability, but also for dealing with partiality issues in \Coq{} and in their traditional role in computable analysis as formalization of computational tasks.

Let us list all the packages that we decided to export from \Incone{} with a short description of their contents.
\begin{itemize}
\item \demph{\Mf{}:} The \Mf{} library contains statements from the theory of multivalued functions.
  The formalization closely follows the description in the first part of Section~\ref{sec: operators}.
  Several useful notations are introduced, for example \verb$_ ->> _$ for multivalued functions and \verb$_ \tightens _$ for the tightening relation.
  An exhaustive list of the important concepts can be found in the preamble of the file \verb$mf.v$ in the library.
\item \demph{\Rlzrs{}:} A library for basic realizability using the question and answer structure described in the beginning of Section~\ref{sec: represented spaces}.
  It defines basic notations such as \verb$_ \is_name_of _$ or \verb$_ \realizes _$ and heavily relies on the \Mf{} library.
\item \demph{\Metric{}:} 
  A classical formalization of metric spaces.
  This library contains results that are independent from the theory of represented spaces.
  It also defines the metric notion of a limit, an efficient limit and proves some statements about continuity.
  The layout of \Metric{} roughly follows that of the Coquelicot \cite{boldo2015coquelicot} library.
  However, the definitions are kept close to the classical mathematical treatment and are thus most similar to the metric spaces that can be found in \Coq{}'s standard library.
\end{itemize}
\begin{figure}
  \centering
  \includegraphics[width=0.9\textwidth]{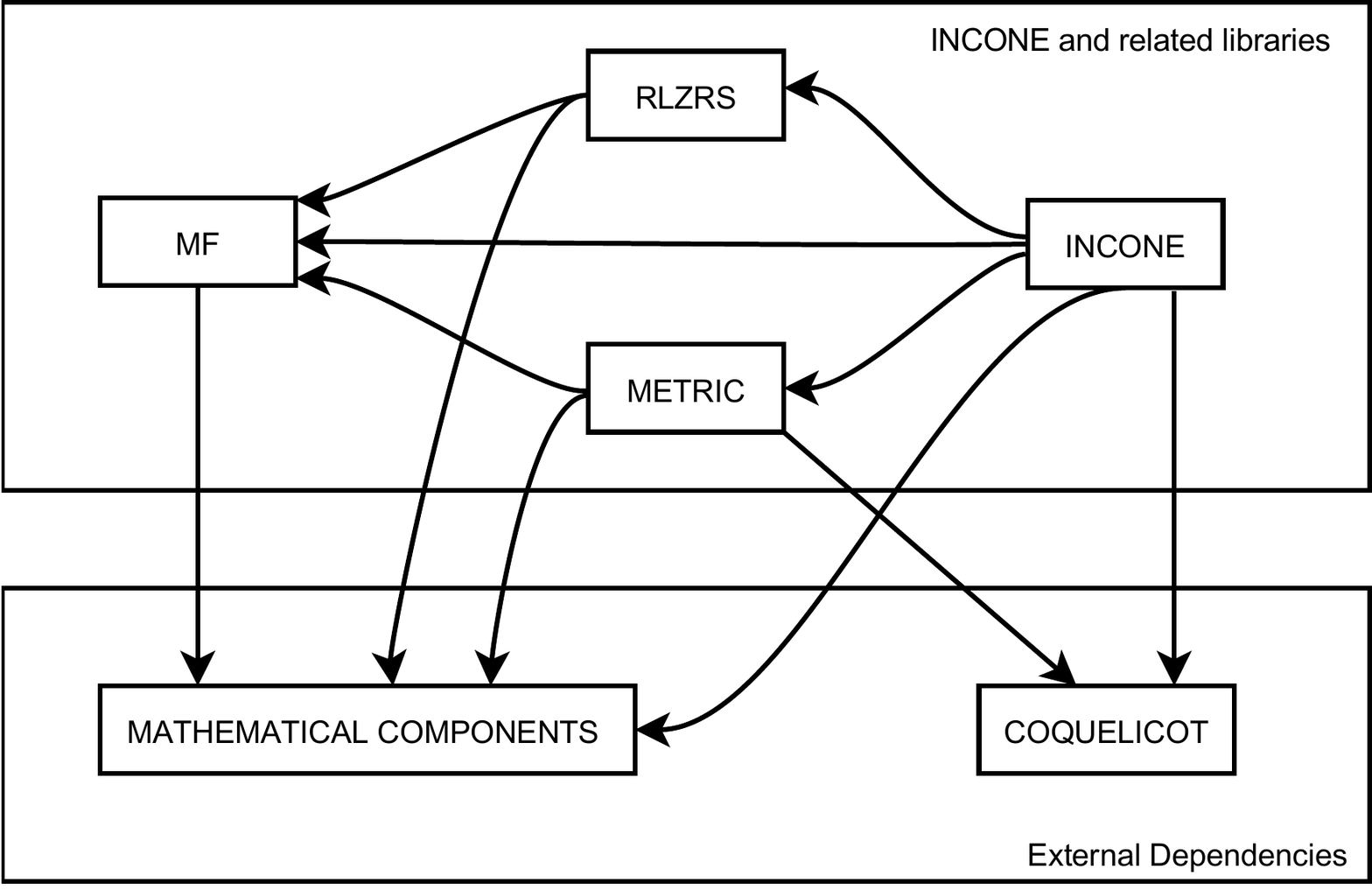}
  \caption{Internal and external dependencies of the incone library.}\label{fig: dependencies}
\end{figure}
The libraries above can, and have to be, installed separately but should be considered part of the \Incone{} system.
However, \Incone{} also has some external dependencies.
As mentioned before it uses the ssreflect proof language.
This language originates from the mathematical components development and \Incone{} also uses other parts of that library \cite{gonthier:inria-00258384}.
Another external dependency that \Incone{} shares with \Metric{} is Coquelicot, a popular formalization of some facts from analysis that is conservative over the axiomatization of the real numbers in the standard library \cite{boldo2015coquelicot}.
See Figure~\ref{fig: dependencies} for an overview of \Incone{}, its sublibraries and dependencies.

The \Mf{} and \Rlzrs{} libraries are fairly manageable and their few files can be navigated by hand and are well documented.
The \Metric{} library is less accessible and heavily interleaved with both the Coquelicot library and the development of metric spaces in \Coq{}'s standard library.
Thus, we take some space to sketch some of its more important features and give additional justification of its existence.

\subsection{The library for metric spaces}

Metric spaces are available in \Coq{}'s standard library and there are several other developments that include a treatment of metric spaces or related concepts.
The original intention of the \Metric{} library was to provide an interface between the the metric spaces from the standard library and a similar concept provided by Coquelicot.
Coquelicot is a popular library for doing analysis in \Coq{} and is conservative over the axiomatization of the real numbers in the standard library.
The concept from Coquelicot that we are referring to is named a ``uniform space'' in Coquelicot.
However, these spaces are more restrictive than what a mathematician would expect to be called a uniform space and instead closely resemble pseudo-metric spaces.
For Coquelicot the choice of pseudo-metric spaces over metric spaces is due to neither \Coq{}, nor the axioms of the real numbers, implying functional extensionality.
This makes it challenging to define a metric on any kind of space of functions.
A pseudo-metric can often be defined in a straightforward manner.
The definitions of limits and continuity used in Coquelicot rely on filters instead of sequences.
As the ``uniform spaces'' are Coquelicot's most general structure and are first-countable, the definitions via filters are equivalent to those using sequences.
However, it is not clear whether the equivalence of the derived concepts can be proven in the setting that Coquelicot works in.

The \Metric{} library provides interfaces with both the standard library of \Coq{} (\verb$MS2M_S$, \verb$M_S2MS$, \verb$Uncv_lim$, \verb$cont_limin$, etc.) and the Coquelicot library (\verb$US2MS$, \verb$MS2US$, \verb$cntp_cntp$, etc.) so that it is possible to reuse results proven there (for instance the lemmas \verb$limD$, \verb$limM$, \verb$R_cmplt$ that assert the limit to commute with additional multiplication and the real numbers to be complete are proven this way).
The definitions of continuity that the metric library uses are as stated in this paper in the chapter about metric spaces.
While these definitions quantify over real numbers $\varepsilon$ and $\delta$, restricting the quantifiers to only reach over a discrete subset of the real numbers often leads to equivalent definitions.
For an easy back and forth between the different convergence statements \Incone{} provides a line of lemmas that contain the phrase ``tpmn'' (for ``two to the power minus $n$'') in their name (\verb$tpmnP$, \verb$lim_tpmn$, \verb$dns_tpmn$, etc.).
For instance \verb$tpmnP$ proves that $2^{-n} \leq 2^{-m}$ on real numbers reflects $m \leq n$ on natural numbers and \verb$lim_tpmn$ says that in the definition of the limit one may replace $\varepsilon$ by $2^{-n}$ and thereby quantification over $\RR$ by quantification over $\NN$.
A similar set of lemmas is provided for replacement of $\varepsilon$ and $\delta$ with rational numbers.
For rational numbers \Incone{} also provides a constructive instantiation of the restriction of the up function (\verb$upQ$, \verb$limQ$, \verb$archimedQ$, etc.).
The function \verb$upQ$, for instance, is useful for recovering computational content from proofs in the standard library, as it may in some cases be used as substitute for the non-computable function \verb$up$ defined on all reals.

Finally, the \Metric{} library proves some results that are not proven in Coquelicot, most of these rely on the axioms that our development assumes over the assumptions made by Coquelicot.
It also defines an interface for handling sub-spaces as continuity of partial functions is important for our purposes.
In the case of metric spaces both continuity and its sequential variant can be recovered from analogous point-wise notions while for represented spaces this is only the case for sequential continuity.
The availability of point-wise notions is not always an advantage and can introduce subtle problems in the treatment of sub-spaces.
Even in the most well-behaved cases there is a difference between a function being continuous in each point of the subset and the restriction of the function being continuous.
For instance, the characteristic function of the closed unit interval has a continuous restriction but is not continuous in either end-point.
This can easily lead to confusion and as a result some of the statements of important theorems about continuous functions on the real numbers from the standard library differ slightly from what a mathematician would expect them to say.
Examples for this include the mean value and the intermediate value theorem.
\end{document}